\documentclass[11pt]{article}

\usepackage{fullpage}
\usepackage{graphicx,amsmath,amssymb,mathtools}
\usepackage{lmodern}
\usepackage{todonotes}

\usepackage{xspace}

\usepackage{algorithm}
\usepackage{algorithmic}
\usepackage{tikz}
\usepackage{ifthen}

\let\oldbibliography\thebibliography
\renewcommand{\thebibliography}[1]{%
  \oldbibliography{#1}%
  \setlength{\itemsep}{0pt}%
}

\title{Testing Formula Satisfaction\thanks{Research supported in part by an ERC-2007-StG grant number 202405.} \thanks{A preliminary version of this work appeared in the Proceedings of 13th Scandinavian Symposium and Workshops on Algorithm Theory (SWAT 2012)}}

\author{Eldar Fischer\thanks{Department of Computer Science, Technion, Haifa 32000, Israel. \mbox{eldar@cs.technion.ac.il}} \and Yonatan Goldhirsh\thanks{Department of Computer Science, Technion, Haifa 32000, Israel. \mbox{jongold@cs.technion.ac.il}} \and Oded Lachish\thanks{Birkbeck, University of London, London, UK. \mbox{oded@dcs.bbk.ac.uk}}}

\date{}

\date{}

\newtheorem{theorem}{Theorem}[section]

\newtheorem{lemma}[theorem]{Lemma}
\newtheorem{corollary}[theorem]{Corollary}

\newtheorem{definition}[theorem]{Definition}
\newtheorem{observation}[theorem]{Observation}

\newcommand{\qed}{\rule{2mm}{2mm}}
\newenvironment{proof}{\par\noindent{\bf Proof.}~}{\hfill  $\qed$}

\newcommand{\SAT}{\textit{SAT}}

\newcommand{\ef}{$\epsilon$-far\xspace}

\newcommand{\et}{$\epsilon$-test\xspace}
\newcommand{\A}{\mathcal{A}}

\newcommand{\rof}{read-once-formula\xspace}

\newcommand{\CH}[1]{\mbox{\tt Children}(#1)\xspace} %
\newcommand{\LC}[1]{h(#1)\xspace} %
\newcommand{\CTL}[2]{\mbox{\tt Critical}_{#1,#2}\xspace} %
\newcommand{\depth}[2]{\mbox{\tt depth}_{#1}(#2)\xspace} %
\newcommand{\outd}[1]{\mbox{\tt out-deg}(#1)\xspace} %
\newcommand{\dist}[2]{\mbox{\tt dist}(#1,#2)\xspace} %
\newcommand{\farness}[3]{\mbox{\tt farness}_{#3}(#1,#2)\xspace} %

\newcommand{\localdist}[1]{2{#1}/3\xspace} %
\newcommand{\twicelocaldist}[1]{4{#1}/3\xspace} %
\newcommand{\orconst}[1]{\lceil 20 {#1}^{-1}\log{{#1}^{-1}}\rceil\xspace} %
\newcommand{\mdepth}[1]{{#1}^{-1}\log{(2{#1}^{-1})}\xspace} 
\newcommand{\numrel}[1]{{#1}^{-2}\log{(2{#1}^{-1})}\xspace} 

\newcommand{\qpower}[1]{-16+16\log{#1}}
\newcommand{\rpower}[1]{-20+16\log{#1}}
\newcommand{\qcomplexity}[1]{{#1}^{\qpower{#1}}\xspace} %
\newcommand{\rqcomplexity}[1]{{#1}^{\rpower{#1}}\xspace} %
\newcommand{\hitprob}[1]{#1/4\xspace} %
\newcommand{\hhitprob}[1]{#1/8\xspace} %
\newcommand{\mDNF}{\mathrm{mDNF}}

\newcommand{\andreduction}[2]{{#1}(1-(2{#2}/{#1})^{-{#2}}/16)}
\newcommand{\andsize}[2]{((2{#2}/{#1})^{-2{#2}}/16)}
\newcommand{\kandsample}[1]{32(2k/{#1})^{2k}}

\newcommand{\genandcnst}{\kandsample{\epsilon}\log(\delta^{-1})}
\newcommand{\slightlysmall}[1]{\andreduction{#1}{k}}
\newcommand{\slightlybig}[1]{{#1}(1+{#1})}
\newcommand{\recurseps}[1]{{#1}(1+(4k/{#1})^{-k})}
\newcommand{\mgqc}[2]{{#1}^{-480(4k/{#1})^{k+3}\log\log({#2}^{-1})}}

\newcommand{\poly}{\mathrm{poly}}

\begin{document}

\maketitle

\setcounter{page}{1}

\begin{abstract}
We study the query complexity of testing for properties defined by read once formulas, as instances of {\em massively parametrized properties}, and prove several testability and non-testability results. First we prove the testability of any property accepted by a Boolean read-once formula involving any bounded arity gates, with a number of queries exponential in $\epsilon$, doubly exponential in the arity, and independent of all other parameters. When the gates are limited to being monotone, we prove that there is an {\em estimation} algorithm, that outputs an approximation of the distance of the input from satisfying the property. For formulas only involving And/Or gates, we provide a more efficient test whose query complexity is only quasipolynomial in $\epsilon$. On the other hand, we show that such testability results do not hold in general for formulas over non-Boolean alphabets; specifically we construct a property defined by a read-once arity $2$ (non-Boolean) formula over an alphabet of size $4$, such that any $1/4$-test for it requires a number of queries depending on the formula size. We also present such a formula over an alphabet of size $5$ that additionally satisfies a strong monotonicity condition.
\end{abstract}

\thispagestyle{empty}
\pagenumbering{arabic}

\section{Introduction}

{\em Property Testing} deals with randomized approximation algorithms that operate under low information situations. The definition of a property testing algorithm uses the following components: A set of {\em objects}, usually the set of strings $\Sigma^*$ over some alphabet $\Sigma$; a notion of a single {\em query} to the input object $w=(w_1,\ldots,w_n)\in\Sigma^*$, which in our case would consist of either retrieving the length $|w|$ or the $i$'th letter $w_i$ for any $i$ specified by the algorithm; and finally a notion of {\em farness}, a normalized distance, which in our case will be the Hamming distance --- $\farness{w}{v}{}$ is defined to be $\infty$ if $|w|\neq|v|$ and otherwise it is $|\{i:w_i\neq v_i\}|/|v|$.

Given a {\em property} $P$, that is a set of objects $P\subseteq \Sigma^*$, an integer $q$, and a farness parameter $\epsilon>0$, an {\em $\epsilon$-test for $P$ with query complexity $q$} is an algorithm that is allowed access to an input object only through queries, and distinguishes between inputs that satisfy $P$ and inputs that are $\epsilon$-far from satisfying $P$ (that is, inputs whose farness from any object of $P$ is more than $\epsilon$), while using at most $q$ queries. By their nature the only possible testing algorithms are probabilistic, with either $1$-sided or $2$-sided error ($1$-sided error algorithms must accept objects from $P$ with probability $1$). Traditionally the query ``what is $|w|$'' is not counted towards the $q$ query limit.

The ultimate goal of Property-Testing research is to classify properties according to their optimal $\epsilon$-test query-complexity. In particular, a property whose optimal query complexity depends on $\epsilon$ alone and not on the length $|w|$ is called {\em testable}. In many (but not all) cases a ``query-efficient'' property test will also be efficient in other computational resources, such as running time (usually it will be the time it takes to retrieve a query multiplied by some function of the number of queries) and space complexity (outside the space used to store the input itself).

Property-Testing was first addressed by Blum, Luby and Rubinfeld~\cite{BlumLR93}, and most of its general notions were first formulated by Rubinfeld and Sudan~\cite{RubinfeldS96}, where the investigated properties are mostly of an algebraic nature, such as the property of a Boolean function being linear. The first excursion to combinatorial properties and the formal definition of testability were by Goldreich, Goldwasser and Ron~\cite{GGR}. Since then Property-Testing has attracted significant attention leading to many results. For surveys see \cite{Fischer01theart}, \cite{Goldreich10}, \cite{Ron2008}, \cite{Ron2010}.

Many times {\em families} of properties are investigated rather than individual properties, and one way to express such families is through the use of parameters. For example, $k$-colorability (as investigated in \cite{GGR}) has an integer parameter, and the more general partition properties investigated there have the sequence of density constraints as parameters. In early investigations the parameters were considered ``constant'' with regards to the query complexity bounds, which were allowed to depend on them arbitrarily. However, later investigations involved properties whose ``parameter'' has in fact a description size comparable to the input itself. Probably the earliest example of this is \cite{Newman02}, where properties accepted by a general read-once oblivious branching program are investigated. In such a setting a general dependency of the query complexity on the parameter is inadmissible, and indeed in \cite{Newman02} the dependency is only on the maximum width of the branching program, which may be thought of as a complexity parameter of the stated problem.

A fitting name for such families of properties is {\em massively parametrized properties}. A good way to formalize this setting is to consider an input to be divided to two parts. One part is the {\em parameter}, the branching program in the example above, to which the testing algorithm is allowed full access without counting queries. The other part is the {\em tested input}, to which the algorithm is allowed only a limited number of queries as above. Also, in the definition of farness only changes to the tested input are allowed, and not to the parameter. In other words, two ``inputs'' that differ on the parameter part are considered to be $\infty$-far from each other. In this setting also other computational measures commonly come into play, such as the running time it takes to plan which queries will be made to the tested input.

Recently, a number of results concerning a massively parametrized setting (though at first not under this name) have appeared. See for example \cite{SOIL,ChakrabortyFLMN07,FischerLNMY08,FischerY11} and the survey~\cite{newman2010}, as well as~\cite{PCPP}, where such an $\epsilon$-test was used as part of a larger mechanism.

A central area of research in Property-Testing in general and Massively-Parametrized Testing in particular is to associate the query complexity of problems to their other measures of complexity. There are a number of results in this direction, to name some examples see~\cite{AlonKNS00,Newman02,FischerNS04}. In~\cite{Ben-SassonHR05} the study of formula satisfiability was initiated. There it was shown that there exists a property that is defined by a $3$-CNF formula and yet has a query complexity that is linear in the size of the input. This implies that knowing that a specific property is accepted by a $3$-CNF formula does not give any information about its query complexity. In~\cite{HalevyLNT07} it was shown that if a property is accepted by a read-twice CNF formula, then the property is testable. Here we continue this line of research.

In this paper we study the query complexity of properties that are accepted by read once formulas. These can be described as computational trees, with the tested input values at the leaves and logic gates at the other nodes, where for an input to be in the property a certain value must result when the calculation is concluded at the root. 

We prove a number of results. Section \ref{sec:prelim} contains preliminaries. First we define the properties we test, and then we introduce numerous definitions and lemmas about bringing the formulas whose satisfaction is tested into a normalized ``basic form''. These are important and in fact implicitly form a preprocessing part of our algorithms. Once the formula is put in a basic form, testing an assignment to the formula becomes manageable.

In Section~\ref{sec:mos-gen} we show the testability of properties defined by formulas involving arbitrary Boolean gates of bounded arity. For such formula involving only monotone gates, we provide an {\em estimation} algorithm in Section~\ref{sec:kestim}, that is an algorithm that not only tests for the property, but with high probability outputs a real number $\eta$ such that the true farness of the tested input from the property is between $\eta-\epsilon$ and $\eta+\epsilon$. In Section \ref{sec:QuasiPoly} we show that when restricted to And/Or gates, we can provide a test whose query complexity is quasipolynomial in $\epsilon$. We supply a brief analysis of the running times of the algorithms in Section~\ref{sec:running}.

On the other hand, we prove in Section~\ref{sec:untest} that these results can not be generalized to alphabets that have at least four different letters. We construct a formula utilizing only one (symmetric and binary) gate type over an alphabet of size $4$, such that the resulting property requires a number of queries depending on the formula (and input) size for a $1/4$-test. We also prove that for the cost of one additional alphabet symbol, we can construct a non-testable explicitly monotone property (both the gate used and the acceptance condition are monotone).

Results such as these might have interesting applications in computational complexity. One interesting implication of the testability results here is that any read-once formula accepting an untestable Boolean property must use unbounded arity gates other than And/Or. By proving that properties defined by formulas of a simple form admit efficient property testers, one also paves a path for proving that certain properties cannot be defined by formulas of a simple form --- just show that these properties cannot be efficiently testable. Since property testing lower bounds are in general easier to prove than computational complexity lower bounds, we hope that this can be a useful approach.

\subsubsection*{Acknowledgment}
We thank Prajakta Nimbhorkar for the helpful discussion during the early stages of this work.

\section{Preliminaries}\label{sec:prelim}

We use $[k]$ to denote the set $\{1,\dots,k\}$.
A {\em digraph} $G$ is a pair $(V,E)$ such that $E\subseteq V\times V$.
For every $v\in V$ we set $\outd{v} = \left|\{u\in V\mid (v,u)\in E\}\right|$.
A {\em path} is a tuple $(u_1,\dots,u_k)\in |V|^k$ such that $u_1,\dots,u_k$ are all distinct 
and $(u_i,u_{i+1})\in E$ for every $i\in[k-1]$.
The {\em length} of a path $(u_1,\dots,u_k)\in |V|^k$ is $k-1$.
We say that there is a path from $u$ to $v$ if there exists a path 
 $(u_1,\dots,u_k)$  in $G$ such that $u_1 = u$, and $u_k=v$.
The {\em distance} from $u\in V$ to $v\in V$, denoted $\dist{u}{v}$, is the length of the shortest path from $u$ to $v$ if one exists
and infinity otherwise.

We use the standard terminology for outward-directed rooted trees. A {\em rooted directed tree} is a tuple $(V,E,r)$, where $(V,E)$ is a digraph, $r\in V$ and for every $v\in V$ there is a unique path from $r$ to $v$. Let $u,v\in V$.
If $\outd{v}=0$ then we call $v$ a leaf.
We say that $u$ is an {\em ancestor} of $v$  and $v$ is a {\em descendant} of $u$
if there is a path from $u$ to $v$.
We say that $u$ is a {\em child} of $v$ and $v$ is a {\em parent} of $u$
if $(v,u)\in E$,  and set $\CH{v} = \{w\in V \mid w~\mbox{ is a child of}~v\}$.

\subsection{Formulas, evaluations and testing}

With the terminology of rooted trees we now define our properties; first we define what is a formula and then we define what it means to satisfy one.

\begin{definition}[Formula]~\label{def:formula}
A {\em Read-Once Formula} is a tuple $\Phi = (V,E,r,X,\kappa,B,\Sigma)$, where $(V,E,r)$ is a rooted directed tree, $\Sigma$ is an alphabet, $X$ is a set of variables (later on they will take values in $\Sigma$), $B\subseteq \bigcup_{k<\infty}\{\Sigma^{k}\mapsto \Sigma\}$ a set of functions over $\Sigma$, and $\kappa:V\rightarrow B\cup X\cup \Sigma$ satisfies the following (we abuse notation somewhat by writing $\kappa_v$ for $\kappa(v)$).
\begin{itemize}
\item
For every leaf $v\in V$ we have that $\kappa_v\in X\cup\Sigma$.
\item
For every $v$ that is not a leaf $\kappa_v \in B$ is a function whose arity is $|\CH{v}|$.
\end{itemize}
In the case where $B$ contains functions that are not symmetric, we additionally assume that for every $v\in V$ there is an ordering of $\CH{v}=(u_1,\ldots,u_k)$.
\end{definition}
In the special case where $\Sigma$ is the binary alphabet $\{0,1\}$, we say that $\Phi$ is {\em Boolean}.
Unless stated otherwise $\Sigma = \{0,1\}$, in which case we shall omit $\Sigma$ from the definition of formulas.
A formula $\Phi = (V,E,r,X,\kappa,B,\Sigma)$ is called {\em read $k$-times} 
if for every $x \in X$ there are at most $k$ vertices $v\in V$, where $\kappa_v \equiv x$.
We call $\Phi$ a {\em \rof} if it is read $1$-times.
A formula $\Phi = (V,E,r,X,\kappa,B,\Sigma)$ is called {\em $k$-ary} if the arity
(number of children) of all its vertices is at most $k$. If a formula is $2$-ary we then call it {\em binary}.
A function $f:\{0,1\}^n\to\{0,1\}$ is {\em monotone} if whenever $x\in \{0,1\}^n$ is such that $f(x)=1$, then for every $y\in \{0,1\}^n$ such that $x\leq y$ (coordinate-wise) we have $f(y)=1$ as well. If all the functions in $B$ are monotone then we say that $\Phi$ is (explicitly) {\em monotone}. We denote $|\Phi|=|X|$ and call it the {\em formula size} (this makes sense for read-once formulas).

\begin{definition}[Sub-Formula]
Let $\Phi = (V,E,r,X,\kappa,B)$ be a formula and $u\in V$.
The formula $\Phi_u = (V_u,E_u,u,X_u,\kappa,B)$, is such that
 $V_u\subseteq V$, with $v\in V_u$ if and only if $\dist{u}{v}$ is finite,
and $(v,w)\in E_u$ if and only if $v,w\in V_u$ and $(v,w)\in E$.
$X_u$ is the set of all $\kappa_v\in X$ such that $v \in V_u$.
If $u\neq r$ then we call $\Phi_u$ a {\em strict sub-formula}.
We define $|\Phi_u|$ to be the number of variables in $V_u$, that is $|\Phi_u| = |X_u|$,
and the {\em weight} of $u$ with respect to its parent $v$ is defined as $|\Phi_u|/|\Phi_v|$.
\end{definition}

\begin{definition}[assignment to and evaluation of a formula]~\label{def:assignment}
An {\em assignment} $\sigma$ to a formula $\Phi = (V,E,r,X,\kappa,B,\Sigma)$ is
a mapping from $X$ to $\Sigma$.
The {\em evaluation} of $\Phi$ given $\sigma$, denoted (abusing notation somewhat) by $\sigma(\Phi)$, is defined as $\sigma(r)$ where $\sigma:V\to\Sigma$ is recursively defined as follows.
\begin{itemize}
\item If $\kappa_v\in\Sigma$ then $\sigma(v)=\kappa_v$.
\item If $\kappa_v\in X$ then $\sigma(v)=\sigma(\kappa_v)$.
\item Otherwise ($\kappa_v\in B$) denote the members of the set $\CH{v}$ by $(u_1,\ldots,u_k)$ and set $\sigma(v)=\kappa_v(\sigma(u_1),\ldots,\sigma(u_k))$.
\end{itemize}
\end{definition}

Given an assignment $\sigma:X\to\Sigma$ and $u\in V$, we let $\sigma_u$ denote its restriction to $X_u$, but whenever there is no confusion we just use $\sigma$ also for the restriction (as an assignment to $\Phi_u$).

For Boolean formulas, we set $\SAT(\Phi=b)$ to be all the assignments $\sigma$ to $\Phi$ such that $\sigma(\Phi) = b$. When $b=1$ and we do not consider the case $b=0$ in that context, we simply denote these assignments by $\SAT(\Phi)$.
If $\sigma \in \SAT(\Phi)$ then we say that $\sigma$ {\em satisfies} $\Phi$.
Let $\sigma_1,\sigma_2$ be assignments to $\Phi$.
We define $\farness{\sigma_1}{\sigma_2}{\Phi}$ to be the relative Hamming distance between the two assignments. That is, $\farness{\sigma_1}{\sigma_2}{\Phi} = |\{x\in X\mid \sigma_1(x)\neq \sigma_2(x)\}|/|\Phi|$.
For every assignment $\sigma$ to $\Phi$ and every subset $S$ of assignments to $\Phi$
we define $\farness{\sigma}{S}{\Phi} = \min \{\farness{\sigma}{\sigma'}{\Phi}\mid \sigma'\in S\}$.
If $\farness{\sigma}{S}{\Phi} > \epsilon$ then $\sigma$ is {\em $\epsilon$-far} from $S$ and otherwise it is {\em $\epsilon$-close} to $S$.

We now have the ingredients to define testing of assignments to formulas in a massively parametrized model. Namely, the formula $\Phi$ is the parameter that is known to the algorithm in advance and may not change, while the assignment $\sigma:X\to\Sigma$ must be queried with as few queries as possible, and farness is measured with respect to the fraction of alterations it requires.

\begin{definition} {\bf[$(\epsilon,q)$-test]}
An {\em $(\epsilon,q)$-test} for $\SAT(\Phi)$ is a randomized algorithm $\A$  with free access to $\Phi$, 
 that given oracle access to an assignment $\sigma$ to $\Phi$ operates as follows.
\begin{itemize}
\item  $\A$ makes at most $q$ queries to $\sigma$
(where on a query $x \in X$ it receives $\sigma_x$ as the answer).
\item If $\sigma \in \SAT(\Phi)$, then $\A$ accepts (returns $1$) with probability at least $2/3$.
\item If $\sigma$ is $\epsilon$-far from $\SAT(\Phi)$, then $\A$ rejects (returns $0$) with probability
at least $2/3$. Recall that $\sigma$ is $\epsilon$-far from $\SAT(\Phi)$ if its relative Hamming distance from every assignment in $\SAT(\Phi)$ is at least $\epsilon$.
\end{itemize}
We say that $\A$ is {\em non-adaptive} if its choice of queries is independent of their values (and may depend only on $\Phi$).
We say that $\A$ has {\em $1$-sided error} if given oracle access to $\sigma \in \SAT(\Phi)$, it accepts (returns $1$) with probability $1$.
We say that $\A$ is an {\em $(\epsilon,q)$-estimator} if it returns a value $\eta$ such that with probability at least $2/3$, $\sigma$ is both $(\eta+\epsilon)$-close and $(\eta-\epsilon)$-far from $\SAT(\Phi)$.
\end{definition}

We can now summarize the contributions of the paper in the following theorem:

\begin{theorem}[Main Theorem]
The following statements all hold for all constant $k$:
\begin{itemize}
\item For any read-once formula $\Phi$ where $B$ is the set of all functions of arity at most $k$ there exists a $1$-sided $(\epsilon,q)$-test for $\SAT(\Phi)$ with $q=\exp(\poly(\epsilon^{-1}))$ (Theorem \ref{thm:AlgMosGen}).
\item For any read-once formula $\Phi$ where $B$ is the set of all monotone functions of arity at most $k$ there exists an $(\epsilon,q)$-estimator for $\SAT(\Phi)$ with $q=\exp(\poly(\epsilon^{-1}))$ (Theorem \ref{thm:AlgGenEst}).
\item For any read-once formula $\Phi$ where $B$ is the set of all conjunctions and disjunctions of any arity there exists an $(\epsilon,q)$-test for $\SAT(\Phi)$ with $q=\epsilon^{O(\log\epsilon)}$ (Corollary \ref{cor:read-once} of Theorem \ref{thm:read-once}).
\item There exists an infinite family of $4$-valued read-once formulas $\Phi$, where $B$ contains one binary function, and an appropriate $b\in\Sigma$, such that there is no non-adaptive $(\epsilon,q)$-test for $\SAT(\Phi=b)$ with $q=O(\depth{}{\Phi})$, and no adaptive $(\epsilon,q)$-test for $\SAT(\Phi)$ with $q=O(\log(\depth{}{\Phi}))$; there also exists such a family of $5$-valued read-once formulas whose gates and acceptance condition are monotone with respect to a fixed order of the alphabet.
(Theorem \ref{thm:untest} and Theorem \ref{thm:untestmon} respectively).
\end{itemize}
\end{theorem}

Note that for the first two items, the degree of the polynomial is linear in $k$.

\subsection{Basic formula simplification and handling}

In the following, unless stated otherwise, our formulas will all be read-once and Boolean. For our algorithms to work, we will need a somewhat ``canonical'' form of such formulas. We say that two formulas $\Phi$ and $\Phi'$ are {\em equivalent} if $\sigma(\Phi)=\sigma(\Phi')$ for every assignment $\sigma:X\to\Sigma$. 

\begin{definition}
A $1$-witness for a boolean function $f:\{0,1\}^n\to \{0,1\}$ is a subset of coordinates $W\subseteq [n]$ which is minimal (by inclusion) amongst subsets for which there exists an assignment $\sigma:W\to\{0,1\}$ such that for every $x\in \{0,1\}^n$ which agrees with $\sigma$ (that is, for all $i\in W$, we have that $x_i=\sigma(i)$) we have that $f(x)=1$.
\end{definition}

Note that a function can have several $1$-witnesses and that a $1$-witness for a monotone function can always use the assignment $\sigma$ that maps all coordinates to $1$.

\begin{definition}\label{def:mDNF}
The $\mDNF$ (monotone disjunctive normal form) of a monotone boolean function $f:\{0,1\}^n\to \{0,1\}$ is a set of terms $T$ where each term in $T$ is a $1$-witness for $f$ and for every $x\in \{0,1\}^n$, $f(x)=1$ if and only if there exists a term $T_j\in T$ such that for all $i\in T_j$, we have that $x_i=1$.
\end{definition}

\begin{observation}
Any monotone boolean function $f:\{0,1\}^n\to \{0,1\}$ has a unique $\mDNF$ $T$.
\end{observation}

\begin{definition}\label{def:forceful}
For $u\in V$, $v\in\CH{u}$ is called \emph{(a,b)-forceful} if $\sigma(v)=a$ implies $\sigma(u)=b$. $v$ is \emph{forceful} if it is {(a,b)-forceful} for some $a,b\in\{0,1\}$.
\end{definition}

For example, for $\wedge$ all children are $(0,0)$-forceful, and for $\vee$ all children are $(1,1)$-forceful. Forceful variables are variables that cause ``Or-like'' or ``And-like'' behavior in the gate.

\begin{definition}\label{def:unforceable}
A vertex $v\in V$ is called \emph{unforceable} if no child of $v$ is forceful.
\end{definition}

\begin{definition}[$k$-$x$-Basic formula]\label{def:kxbasic}
A read-once formula $\Phi$ is {\em $k$-$x$-basic} if it is Boolean, all the functions in $B$ have arity at least $2$ apart from possible negations, the functions of $B$ are either of arity at most $k$ and unforceable, or $\wedge$ or $\vee$ of arity at least $2$, and $\Phi$ satisfies the following. The negations may only have leaves as children, and there is no leaf $v\in V$ such that $\kappa_v\in \{0,1\}$ (i.e. all leaves are variables). No $\wedge$ is a child of a $\wedge$ and no $\vee$ is a child of a $\vee$. Any variable may appear at most once in a leaf, either positively or negated (with a negation as parent).
\end{definition}

The set of variables that appear negated will be denoted by $\neg X$.

\begin{definition}[$k$-Basic formula]\label{def:kbasic}~\label{def:basic}
A read-once formula $\Phi$ is a {\em $k$-basic} formula if it is $k$-$x$-basic and furthermore, all unforceable functions in $B$ are also monotone. If $B$ contains only conjunctions and disjunctions we abbreviate and call the formula {\em basic}.
\end{definition}

\begin{lemma}\label{lem:make-kx-basic}
Every read-once formula $\Phi$ with gates of arity at most $k$ has an equivalent $k$-$x$-basic formula $\Phi '$, possibly over a different set of functions $B$.
\end{lemma}

\begin{proof}
Suppose that for some $u$ that $v\in\CH{u}$ is (a,b)-forceful. If $b=1$ then $\kappa_u$ can be replaced with an $\vee$ gate, where one input of the $\vee$ gate is $v$ if $a=1$ or the negation of $v$ if $a=0$, and the other input is the result of $u$ when fixing $\sigma(\kappa_v)=1-a$. If $b=0$ then $\kappa_u$ can be replaced with an $\wedge$ gate, where one input of the $\wedge$ gate is $v$ if $a=0$ or the negation of $v$ if $a=1$, and the other input is the negation of the gate $u$ when it is assumed that $\sigma(\kappa_v)=a$. After performing this transformation sufficiently many times we have no forceable gates left.

We will now eliminate $\neg$ gates. Any $\neg$ gate in the input or output of a gate which is not $\wedge$ or $\vee$ can be assimilated into the gate. Otherwise, a $\neg$ on the output of an $\vee$ gate can be replaced with an $\wedge$ gate with $\neg$'s on all of its inputs, according to De-Morgan's laws. Also by De-Morgan's laws, a $\neg$ on the output of an $\wedge$ gate can be replaced with an $\vee$ gate with $\neg$'s on all of its inputs.

Finally, any $\vee$ gates that have $\vee$ children can be merged with them, and the same goes for $\wedge$ gates. Now we have achieved an equivalent $k$-$x$-basic formula.
\end{proof}

Note that $\vee$ and $\wedge$ gates are very much forceable.

\begin{observation}
Any formula $\Phi$ which is comprised of only monotone $k$-arity gates has an equivalent $k$-basic formula $\Phi '$.
\end{observation}

This observation follows by inspecting the above proof, and noticing that monotone gates will never produce negations in the process described.

\subsection{Observations about subformulas and farness}

\begin{definition}[heaviest child $\LC{v}$]~\label{def:children}
Let $\Phi = (V,E,r,X,\kappa,B)$ be a formula.
For every $v\in V$ we define $\LC{v}$ to be $v$ if $\CH{v}=\emptyset$, and otherwise to be an arbitrarily selected vertex $u\in \CH{v}$, such that
 $|\Phi_u| = \max\{|\Phi_w|\mid w\in \CH{v}\}$.
\end{definition}

\begin{definition}[vertex depth $\depth{\Phi}{v}$]~\label{def:depth}
Let $\Phi = (V,E,r,X,\kappa,B)$ be a formula.
For every $v\in V$ we define $\depth{\Phi}{v} = \dist{r}{v}$  and
$\depth{}{\Phi} = \max\{\depth{\Phi}{u}\mid u\in V\}$.
\end{definition}

\begin{observation}~\label{ob:kand}\label{ob:gen-and}
Let $v\in V$ be such that either $\kappa_v\equiv \vee$ and $b=0$ or $\kappa_v\equiv\wedge$ and $b=1$, and $\farness{\sigma}{\SAT(\Phi_v=b)}{}\geq \epsilon$. For every $1>\alpha>0$ there exists $S\subseteq\CH{v}$ such that $\sum_{s\in S} |\Phi_s|\geq \epsilon\alpha^2|\Phi|$
and $\farness{\sigma}{\SAT(\Phi_w=b)}{}\geq \epsilon(1-\alpha)$ for every $w \in S$. Furthermore, there exists a child $u\in\CH{v}$ such that $\farness{\sigma}{\SAT(\Phi_u=b)}{}\geq \epsilon$.
\end{observation}

\begin{proof}
Let $T$ be the maximum subset of $\CH{v}$ such that $\Phi_w$ is $\epsilon(1-\alpha)$-far from being evaluated to $b$ for every $w\in T$. If $\sum_{t\in T} |\Phi_t| < \epsilon\alpha^2|\Phi|$ then the distance from having $\Phi_v$ evaluate to $b$ is at most $\epsilon\alpha^2+\epsilon(1-\alpha)(1-\alpha^2)< \epsilon$, which contradicts the assumption.

For the last part, note that if no such child exists then $\Phi_v$ is $\epsilon$-close to being evaluated to $b$.
\end{proof}

\begin{observation}~\label{ob:or}\label{ob:gen-or}
Let $v\in V$ be such that either $\kappa_v\equiv \vee$ and $b=1$ or $\kappa_v\equiv\wedge$ and $b=0$, and $\farness{\sigma}{\SAT(\Phi_v=b)}{}\geq \epsilon$. Further assume that $\Phi$ is $k$-$x$-basic. For every child $u\in\CH{v}$, $|\Phi_u|\geq |\Phi|\epsilon$ and $\farness{\sigma}{\SAT(\Phi_u=b)}{}\geq\slightlybig{\epsilon}$. Furthermore, $\epsilon\leq 1/2$, and for any $u\in \CH{v}\setminus \{\LC{v}\}$, $\farness{\sigma}{\SAT(\Phi_u=b)}{}\geq 2\epsilon$.
\end{observation}

\begin{proof}
First suppose that the weight of some child $u$ is less than $\epsilon$. In this case setting $u$ to $b$ makes the formula $\Phi_v$ evaluate to $b$ by changing less than an $\epsilon$ fraction of inputs, a contradiction.

Since there are at least two children, every child $u$ is of weight at most $1-\epsilon$ and since setting it to $b$ would make $\Phi_v$ evaluate to $b$, it is at least $\slightlybig{\epsilon}$-far from being evaluated to $b$.

For the last part, note that since Since $|\CH{v}| > 1$, there exists $u\in \CH{v}$ such that 
$|\Phi_u| \leq |\Phi_v|/2$. Thus every assignment to $\Phi_v$ is $1/2$-close to an assignment $\sigma'$ by which $\Phi_v$ evaluates to $b$. Also note that any $u\in \CH{v}\setminus \{\LC{v}\}$ satisfies $|\Phi_u|\leq |\Phi_v|/2$, and therefore if $\Phi_u$ were $2\epsilon$-close to being evaluated to $b$, $\Phi_v$ would be $\epsilon$-close to being evaluated to $b$.
\end{proof}

\subsection{Heavy and Light Children in General Gates}

\begin{definition}
Given a $k$-$x$-basic formula $\Phi$, a parameter $\epsilon$ and a vertex $u$, we let $\ell=\ell(u,\epsilon)$ be the smallest integer such that the size of the $\ell$'th largest child of $u$ is less than $|\Phi|(4k/\epsilon)^{-\ell}$ if it exists, and set $\ell=k+1$ otherwise. The {\em heavy} children of $u$ are the $\ell-1$ largest children of $u$, and the rest of the children of $u$ are its {\em light} children.
\end{definition}

\begin{lemma}\label{lem:heavy-bound}
If an unforceable vertex $v$ has a child $u$ such that $|\Phi_v|(1-\epsilon)\leq |\Phi_u|$, then $\sigma$ is both $\epsilon$-close to $\SAT(\Phi_v=1)$ and $\epsilon$-close to $\SAT(\Phi_v=0)$.
\end{lemma}

\begin{proof}
The child is unforceful, and therefore it is possible to change the remaining children to obtain any output value.
\end{proof}

\begin{observation}\label{ob:two-heavies}
If $\kappa_u\not\equiv\wedge$, $\kappa_u\not\equiv\vee$, $\kappa_u\not\in X$ and $\sigma$ is $\epsilon$-far from $\SAT(\Phi_u=b)$, then it must have at least two heavy children.
\end{observation}

\begin{proof}
By the definition of $\ell$, if there is just one heavy child, then $\ell=2$ and the total weight of the light children is strictly smaller than $\epsilon$. Therefore by Lemma~\ref{lem:heavy-bound} there must be more than one heavy child, as otherwise the gate is $\epsilon$-close to both $0$ and $1$.
\end{proof}

\section{Upper Bound for General Bounded Arity Formula}\label{sec:mos-gen}

Algorithm~\ref{alg:MosGen} tests whether the input is $\epsilon$-close to having output $b$ with $1$-sided error, and also receives a confidence parameter $\delta$. The explicit confidence parameter makes the inductive arguments easier and clearer. The algorithm operates by recursively checking the conditions in Observations~\ref{ob:gen-and} and~\ref{ob:gen-or}.

\begin{theorem}\label{thm:AlgMosGen}
$\mbox{Algorithm~\ref{alg:MosGen}}(\Phi,\epsilon,\delta,\sigma)$ always accepts any input that satisfies the formula $\Phi$, and rejects any input far from satisfying $\Phi$ with probability at least $1-\delta$. Its query complexity (treating $k$ and $\delta$ as constant) is always $O(\exp(poly(\epsilon^{-1})))$.
\end{theorem}

\begin{proof}
Follows from Lemma \ref{lem:mosgen-completeness}, Lemma \ref{lem:mosgen-soundness} and Lemma \ref{lem:mosgen-qc} (in that order) below.
\end{proof}

\begin{algorithm}[H]
  \caption{}
  \label{alg:MosGen}
  \textbf{Input:} read-once $k$-$x$-basic formula $\Phi = (V,E,r,X,\kappa)$, parameters $\epsilon,\delta>0,b\in\{0,1\}$, oracle to $\sigma$.\\
  \textbf{Output:} ``true'' or ``false''.

  \begin{algorithmic}[1]
	\STATE {\bf if} $\epsilon>1$ {\bf then return} ``true'' \label{line:MGBigEps}
	\STATE {\bf if} $\kappa_r\in X$ {\bf then return} the truth value of $\sigma(r)=b$\label{line:MGIsVar}
	\STATE {\bf if} $\kappa_r\in \neg X$ {\bf then return} the truth value of $\sigma(r)=1-b$\label{line:MGIsNeg}
	\IF{($\kappa_r\equiv \wedge$ and $b=1$) or ($\kappa_r\equiv\vee$ and $b=0$)}\label{line:MGIsAnd1}
		\STATE $y\longleftarrow\mbox{``true''}$
		\FOR{$i=1$ to $l=\genandcnst$}\label{line:MGIsNegloop}
			\STATE $u \longleftarrow$ a vertex in $\CH{r}$ selected independently at random, where the probability that $w\in \CH{r}$ is selected is $|\Phi_w|/|\Phi|$ \label{line:SmplAnd1}
			\STATE $y \longleftarrow y \wedge \mbox{Algorithm~\ref{alg:MosGen}}(\Phi_u, (\slightlysmall{\epsilon}),\sigma, \delta/2,b)$\label{line:MGRec1}
		\ENDFOR
		\RETURN $y$
	\ENDIF
	\IF{($\kappa_r\equiv\wedge$ and $b=0$) or ($\kappa_r\equiv\vee$ and $b=1$)}\label{line:MGIsAnd0}
		\STATE {\bf if} there exists a child of weight less than $\epsilon$ {\bf then return} ``true'' \label{line:MGAnd0BigKid}\label{line:MGOr1BigKid}
		\STATE $y\longleftarrow\mbox{``false''}$
		\STATE {\bf for all} {$u\in\CH{r}$} {\bf do} $y \longleftarrow y\vee \mbox{Algorithm~\ref{alg:MosGen}}(\Phi_u, (\slightlybig{\epsilon}),\sigma, \epsilon\delta/2,b)$\label{line:MGRec2}
		\RETURN $y$
	\ENDIF
	\STATE {\bf if} there is a child of weight at least $1-\epsilon$ {\bf then return} ``true'' \label{line:MGUnfBigKid}
	\FORALL{$u\in\CH{r}$}
		\STATE $y_u^0 \longleftarrow \mbox{Algorithm~\ref{alg:MosGen}}(\Phi_u, (\recurseps{\epsilon}),\sigma, \delta/2k,0)$\label{line:MGRec3}
		\STATE $y_u^1 \longleftarrow \mbox{Algorithm~\ref{alg:MosGen}}(\Phi_u, (\recurseps{\epsilon}),\sigma, \delta/2k,1)$\label{line:MGRec4}
	\ENDFOR
	\STATE {\bf if} There exists a string $x\in\{0,1\}^k$ such that $\kappa_r$ on $x$ would evaluate to $b$ and for all $u\in\CH{r}$ we have $y_u^{x_u}$ equal to ``true'' {\bf then return} ``true'' {\bf else return} ``false'' \label{line:MGFindAss}
 \end{algorithmic}
\end{algorithm}

\begin{lemma}\label{lem:mosgen-qc-depth}
The depth of recursion in Algorithm~\ref{alg:MosGen} is at most $16(4k/\epsilon)^{k}\log(\epsilon^{-1})$.
\end{lemma}

\begin{proof}
If $\epsilon>1$ then the condition in Line~\ref{line:MGBigEps} is satisfied and the algorithm returns without making any queries.

All recursive calls occur in Lines~\ref{line:MGRec1}, \ref{line:MGRec2}, \ref{line:MGRec3} and \ref{line:MGRec4}. Since $\Phi$ is $k$-$x$-basic, any call with a subformula whose root is labeled by $\wedge$ results in calls to subformulas, each with a root labeled either by $\vee$ or an unforceable gate, and with the same $b$ value (this is crucial since the $b$ value for which $\wedge$ recurses with a smaller $\epsilon$ is the $b$ value for which $\vee$ recurses with a bigger $\epsilon$, and vice-versa). Similarly, any call with a subformula whose root is labeled by $\vee$ results in calls to subformulas, each with a root labeled either by $\wedge$ or an unforceable gate, and with the same $b$ value. Therefore, an increase of two in the depth results in an increase of the farness parameter from $\epsilon$ to at least $(\slightlysmall{\epsilon})(\recurseps{\epsilon})\geq\epsilon (1+(4k/\epsilon)^{-k}/16)$. Thus in recursive calls of depth $16(4k/\epsilon)^{k}\log(\epsilon^{-1})$ the farness parameter exceeds $1$ and the call returns without making any further calls.
\end{proof}

\begin{lemma}\label{lem:mosgen-qc}
Algorithm~\ref{alg:MosGen} uses at most $\mgqc{\epsilon}{\delta}$ queries.
\end{lemma}

\begin{proof}
If $\epsilon>1$ then the condition in Line~\ref{line:MGBigEps} is satisfied and no queries are made. Therefore assume $\epsilon \leq 1$. Observe that in a specific instantiation at most one query is used, either in Line~\ref{line:MGIsVar} or Line~\ref{line:MGIsNeg}. Therefore the number of queries is upper bounded by the number of instantiations of Algorithm~\ref{alg:MosGen}.

In a specific instantiation at most $\genandcnst$ recursive calls are made in total (note that by Line~\ref{line:MGOr1BigKid} there are at most $1/\epsilon$ children in the case of the condition in Line~\ref{line:MGIsAnd0}, and in the case of an unforceable gate there are at most $2k$ recursive calls). Recall that by Lemma~\ref{lem:mosgen-qc-depth} the depth of the recursion is at most $16(4k/\epsilon)^{k}\log(\epsilon^{-1})$.

To conclude, we note that the value of the confidence parameter in all these calls is lower bounded by $\delta\cdot (\epsilon/2k)^{16(4k/\epsilon)^{k}\log(\epsilon^{-1})}\geq\delta\cdot \epsilon^{32(4k/\epsilon)^{k}\log(k\epsilon^{-1})}$. Therefore at most $(32(4k/\epsilon)^{2k}\log(\delta\cdot \epsilon^{-32(2k/\epsilon)^{k}\log(k\epsilon^{-1})})) ^{16(4k/\epsilon)^{k}\log(\epsilon^{-1})}=\epsilon^{-480(4k/\epsilon)^{k+3}\log\log(\delta^{-1})}$ queries are used.
\end{proof}

\begin{lemma}\label{lem:mosgen-completeness}
If $\Phi$ on $\sigma$ evaluates to $b$ then Algorithm~\ref{alg:MosGen} returns ``true'' with probability $1$.
\end{lemma}

\begin{proof}
If $\epsilon>1$ then the condition of Line~\ref{line:MGBigEps} is satisfied and ``true'' is returned correctly. We proceed with induction over the depth of the formula. If $\depth{}{\Phi}=0$ then $\kappa_r\in X\cup \neg X$. If $\kappa_r\in X$ then since $\Phi$ evaluates to $b$, $\sigma(r)=b$, if $\kappa_r\in \neg X$ then $\sigma(r)=1-b$, and the algorithm returns ``true'' correctly.

Now assume that $\depth{}{\Phi}>0$. Obviously, for all $u\in\CH{r}$, we have that $\depth{}{\Phi}>\depth{}{\Phi_u}$ and therefore from the induction hypothesis any recursive call on a subformula that evaluates to $b'$ returns ``true'' with probability $1$.

If $\kappa_r\equiv\wedge$ and $b=1$ or $\kappa_r\equiv\vee$ and $b=0$, then it must be the case that for all $u\in\CH{r}$, $\Phi_u$ evaluates to $b$. By the induction hypothesis all recursive calls will return ``true'' and $y$ will get the value ``true'', which will be returned by the algorithm.

Now assume that $\kappa_r\equiv\wedge$ and $b=0$ or $\kappa_r\equiv\vee$ and $b=1$. Since $\Phi$ evaluates to $b$ then it must be the case that at least for one $u\in\CH{r}$, $\Phi_u$ evaluates to $b$. By the induction hypothesis, the recursive call on that $u$ will return ``true'', and $y$ will get the value ``true'' which will be returned by the algorithm (unless the algorithm already returned ``true'' for another reason).

Lastly, assume that $r$ is an unforceable gate. Since $\Phi$ evaluates to $b$, the children of $r$ evaluate to an assignment $x$ to $\kappa_r$ which evaluates to $b$. By the induction hypothesis, for every $u\in\CH{r}$ the recursive call on $\Phi_u$ with $x_u$ will return ``true'', and thus the assignment $x$ will, in particular, fill the condition in Line~\ref{line:MGFindAss} and the algorithm will return ``true''.
\end{proof}

\begin{lemma}\label{lem:mosgen-soundness}
If $\sigma$ is $\epsilon$-far from getting $\Phi$ to output $b$ then Algorithm~\ref{alg:MosGen} returns ``false'' with probability at least $1-\delta$.
\end{lemma}

\begin{proof}
The proof is by induction over the tree structure, where we partition to cases according to $\kappa_r$ and $b$. Note that $\epsilon\leq 1$.

If $\kappa_r\in X$ or $\kappa_r\in\neg X$ then by Lines~\ref{line:MGIsVar} or \ref{line:MGIsNeg} the algorithm returns ``false'' whenever $\sigma$ does not make $\Phi$ output $b$.

If $\kappa_r\equiv\wedge$ and $b=1$ or $\kappa_r\equiv\vee$ and $b=0$, since $\sigma$ is $\epsilon$-far from getting $\Phi$ to output $b$ then by Observation~\ref{ob:gen-and} we get that there exists $T\subseteq\CH{r}$ for which it holds that $\sum_{t\in T} |\Phi_t|\geq |\Phi|\epsilon\andsize{\epsilon}{k}$ and each $\Phi_t$ is $\slightlysmall{\epsilon}$-far from being evaluated to $b$. Let $S$ be the set of all vertices selected in Line~\ref{line:SmplAnd1}. The probability of a vertex from $T$ being selected is at least $\epsilon\andsize{\epsilon}{k}$. Since this happens at least $\genandcnst$ times independently, with probability at least $1-\delta/2$ we have that $S\cap T\neq \emptyset$. Letting $w\in T\cap S$, the recursive call on it with parameter $\slightlysmall{\epsilon}$ will return ``false'' with probability at least $1-\delta/2$, which will eventually cause the returned value to be ``false'' as required. Thus the algorithm succeeds with probability at least $1-\delta$.

Now assume that $\kappa_r\equiv\wedge$ and $b=0$ or $\kappa_r\equiv\vee$ and $b=1$. Since $\Phi$ is $\epsilon$-far from being evaluated to $b$, Observation \ref{ob:gen-or} implies that all children are of weight at least $\epsilon$, and therefore the conditions of Line~\ref{line:MGOr1BigKid} would not be triggered. Every recursive call on a vertex $v\in\CH{r}$ is made with distance parameter $\slightlybig{\epsilon}$ and so it returns ``true'' with probability at most $\epsilon\delta/2$. Since there are at most $\epsilon^{-1}$ children of $r$, the probability that none returns ``true'' is at least $1-\delta/2$ and in that case the algorithm returns ``false'' successfully.

Now assume that $\kappa_r$ is some unforceable gate. By Observation \ref{lem:heavy-bound}, since $\Phi$ is $\epsilon$-far from being satisfied the condition in Line~\ref{line:MGUnfBigKid} is not triggered. If the algorithm returned ``true'' then it must be that the condition in Line~\ref{line:MGFindAss} is satisfied. If there exists some heavy child $u\in\CH{r}$ such that $y_u^{b}$ is ``true'' and $y_u^{1-b}$ is ``false'', then by Lemma~\ref{lem:mosgen-completeness} the formula $\Phi_{u}$ does evaluate to $b$ and the string in $x$ must be such that $x_u=b$. For the rest of the children of $r$, assuming the calls succeeded, the subformula rooted in each $v$ is $(\recurseps{\epsilon})$-close to evaluate to $x_v$. Since $u$ is heavy, the total weight of $\CH{r}\setminus\{u\}$ is at most $1-(4k/\epsilon)^{-k}$, and thus by changing at most a $(\recurseps{\epsilon})(1-(4k/\epsilon)^{-k})\leq\epsilon$ fraction of inputs we can get to an assignment where $\Phi$ evaluates to $b$.

If all heavy children $u$ are such that both $y_u^{b}$ and $y_u^{1-b}$ are ``true'', then pick some heavy child $u$ arbitrarily. Since $r$ is unforceable, there is an assignment that evaluates to $b$ no matter what the value of $\Phi_u$ is. Take such an assignment $x$ that fits the real value of $\Phi_u$. Note that for every heavy child $v$ we have that $y_v^{x_v}$ is ``true'', and therefore by changing at most an $(\recurseps{\epsilon})$-fraction of the variables in $\Phi_v$ we can get it to evaluate to $x_v$. The weight of $u$ is at least $(4k/\epsilon)^{-\ell+1}$, thus the total weight of the other heavy children is at most $1-(4k/\epsilon)^{-\ell+1}$ and the total weight of the light children is at most $\frac{\epsilon}{4}(4k/\epsilon)^{-\ell}$. So by changing all subformulas to evaluate to the value implied by $x$ we change at most an $(\recurseps{\epsilon})(1-(4k/\epsilon)^{-\ell+1})+ \frac{\epsilon}{4}(4k/\epsilon)^{-\ell}\leq\epsilon$ fraction of inputs and get an assignment where $\Phi$ evaluates to $b$. Note that this $x$ is not necessarily the one found in Line~\ref{line:MGFindAss}.

Thus we have found that finding an assignment $x$ in Line~\ref{line:MGFindAss}, assuming the calls are correct, implies that $\Phi$ is $\epsilon$-close to evaluate to $b$. The probability that all relevant calls to an assignment return ``true'' incorrectly is at most the probability that any of the $2k$ recursive calls errs, which by the union bound is at most $\delta$, and the algorithm will return ``false'' correctly with probability at least $1-\delta$.
\end{proof}

\section{Estimator for monotone formula of bounded arity}\label{sec:kestim}

Algorithm~\ref{alg:GenEst} operates in a recursive manner, and estimates the distance to satisfying the formula rooted in $r$ according to estimates for the subformula rooted in every child of $r$. The algorithm explicitly receives a confidence parameter $\delta$ as well as the approximation parameter $\epsilon$, and should with probability at least $1-\delta$ return a number $\eta$ such that the input is both $(\eta+\epsilon)$-close and $(\eta-\epsilon)$-far from satisfying the given formula. The explicit confidence parameter makes the inductive arguments easier and clearer.

\begin{algorithm}
  \caption{}
  \label{alg:GenEst}
  \textbf{Input:} read-once $k$-basic formula $\Phi = (V,E,r,X,\kappa)$, parameters $\epsilon,\delta>0$, oracle to $\sigma$ .\\
  \textbf{Output:} $\eta\in [0,1]$.

  \begin{algorithmic}[1]
	\STATE {\bf if} $\kappa_r\in X$ {\bf then return}  $1-\sigma(\kappa_r)$\label{line:GenEstQueryOne}
	\STATE {\bf if} $\epsilon > 1$  {\bf then return} $0$\label{line:GenEstEpsilon}
	\STATE {\bf if} $\kappa_r\equiv\vee$ and there exists $u\in\CH{r}$ with $|\Phi_u|<\epsilon|\Phi|$ {\bf then return} $0$\label{line:GenEstSmallOr}
  	\IF{$\kappa_r\equiv \wedge$}\label{line:GenEstAnd}
		\FOR{$i=1$ to $l=\lceil 1000\epsilon^{-2k-2}(4k)^{2k}\cdot\log(1/\delta)\rceil$}\label{line:GenEstSetSelect}
			\STATE $u \longleftarrow$ a vertex in $\CH{r}$ selected independently at random, where the probability that $w\in \CH{r}$ is selected is $|\Phi_w|/|\Phi|$\label{line:GenEstVertexSelect}
			\STATE $\alpha_i \longleftarrow \mbox{Algorithm~\ref{alg:GenEst}}(\Phi_u, \epsilon(1-(4k/\epsilon)^{-k}/8),\delta\epsilon(4k/\epsilon)^{-k}/16,\sigma)$\label{line:GenEstRecursive1}
		\ENDFOR	
		\RETURN $\sum_{i=1}^l\alpha_i/l$
  	\ELSE
		\STATE {\bf for} every light child $u$ of $r$ set $\alpha_u\longleftarrow 0$
		\STATE {\bf for} every heavy child $u$ of $r$ set $\alpha_u\longleftarrow \mbox{Algorithm~\ref{alg:GenEst}}(\Phi_u, \epsilon(1+(4k/\epsilon)^{-k}),\delta/\max\{k,1/\epsilon\},\sigma)$ \label{line:GenEstRecursive2}
		\STATE {\bf for} every term $C$ in the $\mDNF$ of $\kappa_r$ set $\alpha_C \longleftarrow \sum_{u\in C}  \alpha_u\cdot\frac{|\Phi_u|}{|\Phi|}$\label{line:GenTermEval}
		\RETURN $\min\{\alpha_C:C\in\mDNF(\kappa_r)\}$\label{line:GenDNFEval}
  \ENDIF
 \end{algorithmic}
\end{algorithm}

The following states that Algorithm \ref{alg:GenEst} indeed gives an estimation of the distance. While estimation algorithms cannot have $1$-sided error, there is an additional feature of this algorithm that makes it also useful as a $1$-sided test (by running it and accepting if it returns $\eta=0$).

\begin{theorem}~\label{thm:AlgGenEst}
With probability at least $1-\delta$, the output of $\mbox{Algorithm~\ref{alg:GenEst}}(\Phi,\epsilon,\delta,\sigma)$ is an $\eta$ such that the assignment $\sigma$ is both $(\eta+\epsilon)$-close to satisfying $\Phi$ and $(\eta-\epsilon)$-far from satisfying it. Additionaly, if the assignment $\sigma$ satisfies $\Phi$ then $\eta=0$ with probability $1$. Its query complexity (treating $k$ and $\delta$ as constant) is always $O(\exp(poly(\epsilon^{-1})))$.
\end{theorem}
\begin{proof}
The bound on the number of queries is a direct result of Lemma~\ref{lem:GenEstQC} below. Given that, the correctness proof is done by induction on the height of the formula. The base case (for any $\epsilon$ and $\delta$) is the observation that an instantiation of the algorithm that makes no recursive calls (i.e.\ triggers the condition in Line \ref{line:GenEstQueryOne} or \ref{line:GenEstEpsilon}) always gives a value that satisfies the assertion.

The induction step uses Lemma~\ref{lem:GenEstNotAnd} and Lemma~\ref{lem:GenEstAnd} below. Given that the algorithm performs correctly (for any $\epsilon$ and $\delta$) for every formula $\Phi'$ of height smaller than $\Phi$, the assertions of the lemma corresponding to $\kappa_r$ (out of the two) are satisfied, and so the correctness for $\Phi$ itself follows.
\end{proof}

The dependency on $\delta$ can be made into a simple logarithm by a standard amplification technique: Algorithm~\ref{alg:GenEst} is run $O(1/\delta)$ independent times, each with a confidence parameter $2/3$, and then the median of the outputs is taken.

\begin{lemma}\label{lem:GenEstLR}
When called with $\Phi$, $\epsilon$, $\delta$, and oracle access to $\sigma$, Algorithm~\ref{alg:GenEst} goes down at most $2(4k/\epsilon)^k\log(1/\epsilon)=\mathrm{poly}(\epsilon)$ recursion levels. In those recursion levels, $\delta$ decreases by a factor of at most $(\epsilon(4k/\epsilon)^{-k}/16)^{2(4k/\epsilon)^k\log(1/\epsilon)}=\exp(\mathrm{poly}(1/\epsilon))$.
\end{lemma}
\begin{proof}
Recursion can only happen on Line \ref{line:GenEstRecursive1} and Line \ref{line:GenEstRecursive2}. Moreover, because of the formula being $k$-basic, recursion cannot follow through Line \ref{line:GenEstRecursive1} two recursion levels in a row. Therefore, in every two consecutive levels of the recursion $\epsilon$ is increased by a factor of at least $$(1+(4k/\epsilon)^{-k})\cdot(1-(4k/\epsilon)^{-k}/8)\geq (1+\frac34(4k/\epsilon)^{-k}).$$
\end{proof}

\begin{lemma}\label{lem:GenEstQC}
When called with $\Phi$, $\epsilon$, $\delta$, and oracle access to $\sigma$, Algorithm~\ref{alg:GenEst} uses a total of at most $\exp(\mathrm{poly}(1/\epsilon))$ queries.
\end{lemma}
\begin{proof}
Denoting by $\delta'$ the smallest value of $\delta$ in any recursive call, it holds that  $\delta'\geq\delta(\epsilon(4k/\epsilon)^{-k}/16)^{2(4k/\epsilon)^k\log(1/\epsilon)}$ by Lemma \ref{lem:GenEstLR}. The number of recursive calls per instantiation of the algorithm is thus at most $l'=\lceil 1000\epsilon^{-2k-2}(4k)^{2k}\cdot\log(1/\delta')\rceil=\mathrm{poly}(1/\epsilon)$. As the algorithm may make at most one query per instantiation, and this only in the case where a recursive call is not performed, the total number of queries is (bounding the recursion depth through Lemma \ref{lem:GenEstLR}) at most $(l')^{2(4k/\epsilon)^k\log(1/\epsilon)}=\exp(\mathrm{poly}(1/\epsilon))$.
\end{proof}

\begin{lemma}~\label{lem:GenEstNotAnd}
If $\kappa_r\not\equiv\wedge$ and all recursive calls satisfy the assertion of Theorem \ref{thm:AlgGenEst}, then with probability at least $1-\delta$ the current instantiation of Algorithm~\ref{alg:GenEst} provides a value $\eta$ such that $\sigma$ is both $(\eta+\epsilon)$-close to satisfying $\Phi$ and $(\eta-\epsilon)$-far from satisfying it. Furthermore, if $\sigma$ satisfies $\Phi$ then with probability $1$ the output is $\eta=0$.
\end{lemma}
\begin{proof}
First we note that Step \ref{line:GenEstSmallOr}, if triggered, gives a correct value for $\eta$ (as the $\sigma$ can be made into a satisfying assignment by changing possibly all variables of the smallest child of $r$). We also note that if $\kappa_r\equiv\vee$ and Step \ref{line:GenEstSmallOr} was not triggered, then by definition all of $r$'s children are heavy, and there are no more than $1/\epsilon$ of them.

The true farness of $\sigma$ from $\Phi$ is the minimum over all terms $C$ in $\kappa_r$ of the adjusted cost of making all children of $C$ evaluate to $1$, which is $\sum_{u\in C}\farness{\sigma}{\SAT(\Phi_u)}{}\cdot\frac{|\Phi_u|}{|\Phi|}$. Now in this case there are clearly no more than $\max\{k,\epsilon^{-1}\}$ children, and so by the union bound, with probability at least $1-\delta$, every call done through Line \ref{line:GenEstRecursive1} gave a value $\eta_u$ so that indeed $\sigma$ is $(\eta_u+\epsilon(1+(4k/\epsilon)^{-k}))$-close and $(\eta_u-\epsilon(1+(4k/\epsilon)^{-k}))$-far from $\Phi_u$.

Now let $D_i$ denote $C_i$ minus any light children that it may contain. It may be that some $D_i$'s contain all heavy children, but as there are no forcing children (and there are heavy children) it must be the case that some $D_i$'s do not contain all heavy children, and in Line \ref{line:GenDNFEval} these will dominate. Note that $\sum_{u\in D_i}|\Phi_{u}|\leq (1-(4k/\epsilon)^{1-\ell})|\Phi|$ for any $D_i$ not containing a heavy child. This implies by bounding $(1+(4k/\epsilon)^{-k}))\cdot(1-(4k/\epsilon)^{1-\ell})$:

$$\sum_{u\in D_i}\!\! \farness{\sigma}{\SAT(\Phi_u)}{}\cdot \frac{|\Phi_u|}{|\Phi|}-\epsilon < \frac{\sum_{u\in D_i}\eta_u|\Phi_{u}|}{|\Phi|}$$ 
$$ < \sum_{u\in D_i}\!\! \farness{\sigma}{\SAT(\Phi_u)}{}\cdot \frac{|\Phi_u|}{|\Phi|} +\epsilon -2k(4k/\epsilon)^{-\ell}$$

Now the true farness of $C_i$ not containing all heavy children is at least that of $D_i$, and at most that of $D_i$ plus with the added farness of making all light children evaluate to $1$, which is bounded by $k(4k/\epsilon)^{-\ell}$. This means that for such a $C_i$ we have:

$$\sum_{u\in C_i}\!\! \farness{\sigma}{\SAT(\Phi_u)}{}\cdot \frac{|\Phi_u|}{|\Phi|}-\epsilon < \frac{\sum_{u\in D_i}\eta_u|\Phi_{u}|}{|\Phi|}$$
$$ < \sum_{u\in C_i}\!\! \farness{\sigma}{\SAT(\Phi_u)}{}\cdot \frac{|\Phi_u|}{|\Phi|} +\epsilon -k(4k/\epsilon)^{-\ell}$$

The value returned as $\eta$ is the minimum over terms $C_i$ in $\kappa_r$ of $\eta_u \cdot \frac{\sum_{u\in D_i}|\Phi_{u}|}{|\Phi|}$. We also know that this minimum is reached by some $C_j$ which does not contain all heavy children, but it may be that in fact $\farness{\sigma}{\SAT(\Phi)}{}=\sum_{u\in C_i}\!\! \farness{\sigma}{\SAT(\Phi_u)}{}\cdot \frac{|\Phi_u|}{|\Phi|}$ for some $i\neq j$ (the true farness is the minimum of the total farness of each clause, but it may be reached by a different clause).

By our assumptions
$$\farness{\sigma}{\SAT(\Phi)}{}-\epsilon= \sum_{u\in C_i}\!\! \farness{\sigma}{\SAT(\Phi_u)}{}\cdot \frac{|\Phi_u|}{|\Phi|}-\epsilon$$
$$\leq \sum_{u\in C_j}\!\! \farness{\sigma}{\SAT(\Phi_u)}{}\cdot \frac{|\Phi_u|}{|\Phi|}-\epsilon<\eta$$
so we have one side of the required bound. For the other side, we split into cases. If $C_i$ also does not contain all heavy children then we use the way we calculated $\eta$ as the minimum over the corresponding sums:
$$\eta=\frac{\sum_{u\in D_j}\eta_u|\Phi_{u}|}{|\Phi|}\leq\frac{\sum_{u\in D_i}\eta_u|\Phi_{u}|}{|\Phi|}<\farness{\sigma}{\SAT(\Phi)}{}+\epsilon$$
In the final case, we note that by the assumptions on the light children we will always have (recalling that $C_i$ will in particular have all heavy children of $C_j$):
$$\eta=\frac{\sum_{u\in D_j}\eta_u|\Phi_{u}|}{|\Phi|}<\sum_{u\in C_j}\!\! \farness{\sigma}{\SAT(\Phi_u)}{}\cdot \frac{|\Phi_u|}{|\Phi|} +\epsilon -k(4k/\epsilon)^{-\ell}$$ 
$$\leq \sum_{u\in C_i}\!\! \farness{\sigma}{\SAT(\Phi_u)}{}\cdot \frac{|\Phi_u|}{|\Phi|} +\epsilon$$

where the rightmost term equals $\farness{\sigma}{\SAT(\Phi)}{}+\epsilon$ as required.

For the last part of the claim, note that if $\sigma$ satisfies $\Phi$, then in particular, one of the terms $C$ of $\kappa_r$ must be satisfied. By the induction hypothesis, for all $u\in C$ we would have $\alpha_u=0$ and therefore $\alpha_C=0$, and since $\alpha$ is taken as a minimum over all terms we would have $\alpha=0$.
\end{proof}

\begin{lemma}~\label{lem:GenEstAnd}
If $\kappa_r\equiv\wedge$ and all recursive calls satisfy the assertion of Theorem \ref{thm:AlgGenEst}, then with probability at least $1-\delta$ the current instantiation of Algorithm~\ref{alg:GenEst} provides a value $\eta$ such that $\sigma$ is both $(\eta+\epsilon)$-close to satisfying $\Phi$ and $(\eta-\epsilon)$-far from satisfying it. If $\sigma$ satisfies $\Phi$ then with probability $1$ the output is $\eta=0$.
\end{lemma}
\begin{proof}
First note that if we sample a vertex $w$ according to the distribution of Line \ref{line:GenEstSetSelect} and then take the true farness $\farness{\sigma}{\SAT(\Phi_w)}{}$, then the expectation (but not the value) of this equals $\farness{\sigma}{\SAT(\Phi)}{}$. This is because to make $\sigma$ evaluate to $1$ at the root, we need to make all its children evaluate to $1$, an operation whose adjusted cost is given by the weighted sum of farnesses that corresponds to the expectation above.

Thus, denoting by $X_i$ the random variable whose value is $\farness{\sigma}{\SAT(\Phi_{w_i})}{}$ where $w_i$ is the vertex picked in the $i$th iteration, we have $\mathrm{E}[X_i]=\farness{\sigma}{\SAT(\Phi)}{}$. By a Chernoff type bound, with probability at least $1-\delta/2$, the average $X$ of $X_1,\ldots,X_l$ is no more than $\epsilon^{k+1}(4k)^{-k}/16$ away from $\mathrm{E}[X_i]$ and hence satisfies:
$$\farness{\sigma}{\SAT(\Phi)}{}-\epsilon^{k+1}(4k)^{-k}/16 < X < \farness{\sigma}{\SAT(\Phi)}{}+\epsilon^{k+1}(4k)^{-k}/16$$

Then note that by the Markov inequality, the assertion of the lemma means that with probability at least $1-\delta/2$, all calls done in Line \ref{line:GenEstRecursive2} but at most $\epsilon(4k/\epsilon)^{-k}/16$ of them return a value $\eta_w$ so that $\sigma$ is $(\eta_w+\epsilon(1-(4k/\epsilon)^{-k}/8))$-close and $(\eta_w-\epsilon(1-(4k/\epsilon)^{-k}/8))$-far from $\Phi_w$.

When this happens, at least $(1-\epsilon(4k/\epsilon)^{-k}/16)$ of the answers $\alpha_i$ of the calls are up to $\epsilon(1-(4k/\epsilon)^{-k}/16))$ away from each corresponding $X_i$, and at most $\epsilon(4k/\epsilon)^{-k}/16$ of the answers $\alpha_i$ are up to $1$ away from each corresponding $X_i$. Summing up these deviations, the final answer average $\eta$ satisfies
$$X-\epsilon(1-(4k/\epsilon)^{-k}/4)-\epsilon(4k/\epsilon)^{-k}/16 < \eta < X+\epsilon(1-(4k/\epsilon)^{-k}/4)+\epsilon(4k/\epsilon)^{-k}/16$$

With probability at least $1-\delta$ both of the above events occur, and summing up the two inequalities we obtain the required bound
$$\farness{\sigma}{\SAT(\Phi)}{}-\epsilon \leq \eta < \farness{\sigma}{\SAT(\Phi)}{}+\epsilon$$
\end{proof}

\section{Quasi-polynomial Upper Bound for Basic-Formulas}\label{sec:QuasiPoly}
Let $\Phi = (V,E,r,X,\kappa,B)$ be a basic formula and $\sigma$ be an assignment to $\Phi$.

The main idea of the algorithm is to randomly choose a full root to leaf path, and recurs over all the children of ``$\vee$'' vertices on this path that go outside of it, if they are not too many. The main technical part is in proving that if $\sigma$ is indeed $\epsilon$-far from satisfying $\Phi$, then many of these paths have few such children (few enough to recurs over all of them), where additionally the distance of $\sigma$ from satisfying the corresponding sub-formulas is significantly larger.
An interesting combinatorial corollary of this is that formulas, for which there are not a lot of leaves whose corresponding paths have few such children, do not admit $\epsilon$-far assignments at all.

\subsection{Critical and Important}

To understand the intuition behind the following definitions, it is useful to first consider what happens if
we could locate a vertex that is ``$(\epsilon,\sigma)$-critical'' in the sense that is defined next.

\begin{definition}~\label{def:important-critical} {\bf [~{\em $(\epsilon,\sigma)$-important, $(\epsilon,\sigma)$-critical}~]}
A vertex $v\in V$ is {\em $(\epsilon,\sigma)$-important} if  $\sigma \notin \SAT(\Phi)$, and 
for every $u$ that is either $v$ or an ancestor of $v$, we have that
\begin{itemize}
\item $\farness{\sigma}{\SAT(\Phi_u)}{}\geq (\localdist{\epsilon})(1+\localdist{\epsilon})^{\lfloor \depth{\Phi}{u}/3\rfloor}$
\item If $\kappa_u \equiv \vee$ and $u\neq v$ then $\LC{u}$ is either $v$ or an ancestor of $v$.
\end{itemize}
An {\em $(\epsilon,\sigma)$-critical} vertex $v$ is an $(\epsilon,\sigma)$-important vertex $v$ for which $\kappa_v\in X$.
\end{definition}

Note that such a vertex is never too deep,
since $\farness{\sigma}{\SAT(\Phi_u)}{}$ is always at most $1$. The following observation follows from Definition~\ref{def:important-critical}.

\begin{observation}~\label{prop:important-notdeep}
If $v$ is $(\epsilon,\sigma)$-important, then $\depth{\Phi}{v} \leq 4\mdepth{\epsilon}$.
\end{observation}

A hypothetical oracle that provides a critical vertex can be used as follows.
If $v$ is the vertex returned by such an oracle, then for every ancestor $u$ of
$v$, such that $\kappa_u = \vee$, and every $w\in \CH{v}$ that is not an ancestor of $v$, a number of recursive calls with $\Phi_w$ and distance parameter significantly larger than $\epsilon$ are used.
The following Lemma implies that if for each of these vertices one of the recursive calls returned $0$, then we know that $\sigma\not\in \SAT(\Phi)$.

\begin{definition}[Special relatives]
The set of special relatives of $v\in V$ is the set $T$ of every  
$u$ that is not an ancestor of $v$ or $v$ itself but is a child of an ancestor $w$ of $v$, where $\kappa_w \equiv \vee$.  
\end{definition}

\begin{lemma}~\label{prop:witness}
If $\sigma\not\in \SAT(\Phi_u)$ for every $u\in T\cup\{v\}$, then $\sigma\not\in \SAT(\Phi)$.
\end{lemma}

\begin{proof}
If $\depth{\Phi}{v} = 0$ then $\sigma\not\in \SAT(\Phi_v)$ implies
$\sigma\not\in \SAT(\Phi)$.
Assume by induction that the lemma holds for any
formula $\Phi'= (V',E',r',X',\kappa')$, assignment $\sigma'$ to $\Phi'$
and vertex $u\in V'$ such that $0\leq \depth{\Phi'}{u} < \depth{\Phi}{v}$.
Let $w$ be the parent of $v$. 
Observe that the special relatives of $w$ are a subset of the special relatives of
$v$ and hence by the induction assumption we only need to prove that
 $\sigma\not\in \SAT(\Phi_w)$ in order to infer that $\sigma\not\in \SAT(\Phi)$.

If $\kappa_w \equiv \wedge$, then $\sigma\not\in \SAT(\Phi_v)$ implies that
$\sigma\not\in \SAT(\Phi_w)$.
If $\kappa_w \equiv \vee$, then $\sigma\not\in \SAT(\Phi_v)$
and $\sigma\not\in \SAT(\Phi_u)$ for every $u\in T$ implies that 
$\sigma\not\in \SAT(\Phi_w)$, since we have that $\CH{w}\setminus\{v\}\subseteq T$.
\end{proof}

The following lemma states that if $\sigma$ is \ef from $\SAT(\Phi)$, then
 $(\epsilon,\sigma)$-critical vertices are abundant, and so we can locate one of them
by merely sampling a sufficiently large (linear in $1/\epsilon$) number of vertices. 

The main part of the proof that this holds is in showing that if $\sigma$ is only $\localdist{\epsilon}$-far from $\SAT(\Phi)$, then there exists an  $(\epsilon,\sigma)$-critical vertex for $\sigma$. We first show
that this is sufficient to show the claimed abundance of $(\epsilon,\sigma)$-critical vertices,
and then state and prove the required lemma.

\begin{lemma}~\label{lem:many-critical}
If $\sigma$ is $\epsilon$-far from $\SAT(\Phi)$,
then $|\{v | v \mbox{~is~}(\epsilon,\sigma)\mbox{-critical}\}|\geq \hitprob{\epsilon|\Phi|}$.
\end{lemma}

\begin{proof}
Set $\CTL{\epsilon}{\sigma} = \{v | v \mbox{~is~}(\epsilon,\sigma)\mbox{-critical}\}$ and assume on the
contrary that $|\CTL{\epsilon}{\sigma}|<\hitprob{\epsilon|\Phi|}$.
Set $\sigma'$ to be an assignment to $X$ so that for every $s\in V$
where $\kappa_s \in X$, we have that
$\sigma'(\kappa_s) = 1$ if $\kappa_s\in \CTL{\epsilon}{\sigma}$ and otherwise
$\sigma'(x) =\sigma(x)$.
Thus $\CTL{\epsilon}{\sigma'}=\emptyset$.
By the triangle inequality we have that 
$$\farness{\sigma}{\SAT(\Phi)}{}- \farness{\sigma'}{\SAT(\Phi)}{} \leq \farness{\sigma'}{\sigma}{}.$$
Finally, since $\CTL{\epsilon}{\sigma'}=\emptyset$,  Lemma~\ref{lem:critical-exists}, which we prove below, asserts that $\farness{\sigma'}{\SAT(\Phi)}{}<\localdist{\epsilon}$ and we reach a contradiction.
\end{proof}

\begin{lemma}~\label{lem:critical-exists}
If there is no $(\epsilon,\sigma)$-critical vertex, then
 $\sigma$ is $\localdist{\epsilon}$-close to $\SAT(\Phi)$.
\end{lemma}
\begin{proof}
We shall show that if $\sigma$ is $\localdist{\epsilon}$-far from $\SAT(\Phi)$,
then there exists an $(\epsilon,\sigma)$-critical vertex.
Assume that  $\sigma$ is $\localdist{\epsilon}$-far from $\SAT(\Phi)$.
This implies that $r$ is an $(\epsilon,\sigma)$-important vertex.
Hence an $(\epsilon,\sigma)$-important vertex exists.
Let $v$ be an $(\epsilon,\sigma)$-important vertex
 such that $\depth{\Phi}{v}$ is maximal.
Consequently, none of the vertices in $\CH{v}$ is $(\epsilon,\sigma)$-important. 
We next prove that $v$ is $(\epsilon,\sigma)$-critical. 

Assume on the contrary that $v$ is not $(\epsilon,\sigma)$-critical.
Consequently $\kappa_v\not\in X$ and hence to get a contradiction it is sufficient to show that there exists an  $(\epsilon,\sigma)$-important vertex in $\CH{v}$.
If $\kappa_v \equiv \vee$, then by Observation~\ref{ob:or} we get that 
$$\farness{\sigma}{\SAT(\Phi_{\LC{v}})}{} \geq  (\localdist{\epsilon})(1+\localdist{\epsilon})^{\lfloor \depth{\Phi}{\LC{v}}/3\rfloor},$$
and hence $\LC{v}$ is $(\epsilon,\sigma)$-important.

Assume that $\kappa_v \equiv \wedge$.
Let $u$ be such that $\farness{\sigma}{\SAT(\Phi_{u})}{} \geq \farness{\sigma}{\SAT(\Phi_{v})}{}$.
Observation~\ref{ob:gen-and} asserts that such a vertex exists.
We assume that $\depth{\Phi}{u} > 2$, since otherwise it cannot be the case that
$\farness{\sigma}{\SAT(\Phi_{u})}{} < (\localdist{\epsilon})(1+\localdist{\epsilon})^0$.
Let $w\in V$ be the parent of $v$. Since $w$ is an ancestor of $v$ it is $(\epsilon,\sigma)$-important, and hence $\farness{\sigma}{\SAT(\Phi_{w})}{}  \geq (\localdist{\epsilon})(1+\localdist{\epsilon})^{\lfloor \depth{\Phi}{w}/3\rfloor}$. Since $\Phi$ is basic we have that $\kappa_w \equiv \vee$. Thus by Observation~\ref{ob:or} we get that
$$\farness{\sigma}{\SAT(\Phi_{v})}{} \geq  (\localdist{\epsilon})(1+\localdist{\epsilon})^{1+\lfloor \depth{\Phi}{w}/3\rfloor}.$$ Finally since $\farness{\sigma}{\SAT(\Phi_{u})}{} \geq \farness{\sigma}{\SAT(\Phi_{v})}{}$ and additionally we have $\depth{\Phi}{u} = \depth{\Phi}{w}+2$ we get that $$\farness{\sigma}{\SAT(\Phi_{u})}{}  \geq (\localdist{\epsilon})(1+\localdist{\epsilon})^{\lfloor \depth{\Phi}{u}/3\rfloor}.$$
\end{proof}

\subsection{Algorithm}
This algorithm detects far inputs with probability $\Theta(\epsilon)$, but this can be amplified to $2/3$ using iterated applications.
\begin{algorithm}[H]
  \caption{}\label{alg:MonQuasi}
  \textbf{Input:} read-once basic formula $\Phi = (V,E,r,X,\kappa)$, a parameter $\epsilon>0$, oracle to $\sigma$ .\\
  \textbf{Output:} $z\in \{0,1\}$.
  \begin{algorithmic}[1]
\STATE {\bf if } {$\epsilon > 1$} {\bf then return} $1$\label{line:QuasiEpsilon}\label{line:QuasiAVariable}
\STATE {\bf if }  {$\kappa_r\in X$}  {\bf then return}  $\sigma(\kappa_r)$\label{line:QuasiQueryOne}
\STATE\label{line:QuasiPick} Pick $s$ uniformly at random from all $v$ such that $\kappa_v\in X$
\STATE\label{line:ancestors} $A \longleftarrow$ all ancestors $v$ of $s$ such that
$\kappa_v \equiv \vee$~\label{line:QuasiA}
\STATE\label{line:initialize} $R \longleftarrow \left(\bigcup_{v\in A} \CH{v}\right)\setminus \{w\mid w \mbox{ is an ancestor of }s\}$~\label{line:QuasiR}
\STATE {\bf if }  {$|R| > 3\numrel{\epsilon}$} {\bf then return} $1$\label{line:QuasiTooMany}
\FORALL {$u \in R$} \label{line:QuasiCalls}
\STATE $y_u \longleftarrow 1$
\STATE {\bf for} {$i = 1 \mbox{ to } \orconst{\epsilon}$} {\bf do} $y_u \longleftarrow y_u \wedge \mbox{Algorithm~\ref{alg:MonQuasi}}(\Phi_{u},\sigma,\twicelocaldist{\epsilon})$ \label{line:QuasiFor}\label{line:QuasiRecursiveCalls} 
\ENDFOR
\RETURN $\sigma(\kappa_s) \vee \bigvee_{u\in R} y_u$\label{line:QuasiQueryTwo}
 \end{algorithmic}
\end{algorithm}

We can now proceed to prove the correctness of Algorithm~\ref{alg:MonQuasi}. Algorithm~\ref{alg:MonQuasi} is clearly non-adaptive. We next prove that it always returns ``$1$'' for
an assignment that satisfies the formula, and returns ``$0$'' with probability linear in $\epsilon$
for an assignment that is $\epsilon$-far from satisfying the formula. Using $O(1/\epsilon)$ independent
iterations amplifies the later probability to $2/3$.

\begin{lemma}~\label{prop:query-complexity}
For $\epsilon > 0$, Algorithm~\ref{alg:MonQuasi} halts
after using at most $\qcomplexity{\epsilon}$ queries, when called with $\Phi$, $\epsilon$ and oracle access to $\sigma$.
\end{lemma}
\begin{proof}
The proof is formulated as an inductive argument over the value of the (real) farness parameter $\epsilon$. However, it is formulated in a way that it can be viewed as an inductive argument over the integer valued $\lceil\log(\alpha\epsilon^{-1})\rceil$, for an appropriate global constant $\alpha$.

If $\epsilon > 1$, then the condition in Line~\ref{line:QuasiEpsilon} is satisfied,
and there are no queries or recursive calls.
Hence we assume that $\epsilon \leq 1$.
Observe that in a specific instantiation at most one query is used, since 
a query is only made on Line~\ref{line:QuasiQueryOne} or on
Line~\ref{line:QuasiQueryTwo}, and always as part of a ``return'' command.
Hence the number of queries is upper bounded by the number of calls to Algorithm~\ref{alg:MonQuasi} (initial and recursive).
We shall show that the number of these calls is at most $\qcomplexity{\epsilon}$.

Assume by induction that for some $\eta\leq 1$, for every $\eta\leq \eta' \leq 1$, every formula $\Phi'$ and assignment $\sigma'$ to $\Phi'$, on call to Algorithm~\ref{alg:MonQuasi} with $\Phi'$, $\eta'$ and an oracle to $\sigma'$,  at most $\qcomplexity{\eta'}$ calls to Algorithm~\ref{alg:MonQuasi} are made (including recursive ones). 

Assume that $\epsilon > 3\eta/4$.
If $\kappa_r\in X$, then the condition on Line~\ref{line:QuasiAVariable} is satisfied and hence there are no recursive calls. Thus Algorithm~\ref{alg:MonQuasi} is called only once and $1 \leq \qcomplexity{\epsilon}$. 

Assume that $\kappa_r\not\in X$.
Note that every recursive call is done by Line~\ref{line:QuasiRecursiveCalls}.
By Line~\ref{line:QuasiCalls} and  Line~\ref{line:QuasiFor} at most
$|R|\cdot\orconst{\epsilon}$ recursive calls are done.
The condition on Line~\ref{line:QuasiTooMany} ensures that $|R|\cdot\orconst{\epsilon}\leq 3\numrel{\epsilon}\cdot\orconst{\epsilon}$.
According to Line~\ref{line:QuasiRecursiveCalls} each one of these recursive calls is done with distance parameter $4\epsilon/3> \eta$.
Thus by the induction assumption the number of calls to Algorithm~\ref{alg:MonQuasi} is at most $$1+3\numrel{\epsilon}\cdot\orconst{\epsilon}\cdot\qcomplexity{(\twicelocaldist{\epsilon})}.$$
This is less than $\qcomplexity{\epsilon}$.
\end{proof}

The following theorem will be immediate from Lemma \ref{prop:query-complexity} above when coupled
with Lemma \ref{prop:completeness} and Lemma \ref{lem:QuasiSoundness} below.

\begin{theorem}\label{thm:read-once}
Let $\epsilon > 0$.
When Algorithm~\ref{alg:MonQuasi} is called with
$\Phi$, $\epsilon$ and an oracle to $\sigma$, it
uses at most  $\qcomplexity{\epsilon}$ queries; 
if $\sigma\in \SAT(\Phi)$ then it always returns $1$, and if  
$\sigma$ is $\epsilon$-far from $\SAT(\Phi)$
then it returns $0$ with probability at least $\hhitprob{\epsilon}$. 
\end{theorem}

Theorem~\ref{thm:read-once} does not imply that Algorithm~\ref{alg:MonQuasi} is an $\epsilon$-test for $\SAT(\Phi)$.
However it does imply that in order to get an $\epsilon$-test for $\SAT(\Phi)$
 it is sufficient to do the following. 
Call Algorithm~\ref{alg:MonQuasi} repeatedly $\orconst{\epsilon}$ times,
 return $0$ if any of the calls returned $0$, and otherwise return $1$.
This only increases the query complexity to the value in the following corollary.

\begin{corollary}\label{cor:read-once}
There exists an \et for $\Phi$, that uses at most $\rqcomplexity{\epsilon}$ queries.
\end{corollary}

\begin{lemma}~\label{prop:completeness}
Let $\epsilon > 0$ and $\sigma\in \SAT(\Phi)$. Algorithm~\ref{alg:MonQuasi}
returns $1$ when called with $\Phi$, $\epsilon$ and an oracle to $\sigma$.
\end{lemma}
\begin{proof}
To prove the lemma we shall show that if Algorithm~\ref{alg:MonQuasi}
returns $0$, when called with $\Phi$, $\epsilon$ and oracle access to $\sigma$, then $\sigma \not\in \SAT(\Phi)$.
If $\depth{}{\Phi} = 0$ then the condition in Line~\ref{line:QuasiAVariable} is satisfied and $\sigma(\kappa_r)$ is returned.
Hence $\sigma(\kappa_r) = 0$ and therefore $\sigma \not\in \SAT(\Phi)$.
Assume that for every $\epsilon'>0$, $\Phi'$ where $\depth{}{\Phi'} < \depth{}{\Phi}$, and assignment $\sigma'$ to $\Phi'$, if Algorithm~\ref{alg:MonQuasi}
returns $0$, when called with $\Phi'$, $\epsilon'$ and oracle access to $\sigma'$, then $\sigma' \not\in \SAT(\Phi)$.

Observe that the only other way a $0$ can be returned is through Line~\ref{line:QuasiQueryTwo}, if it is reached.
Let $R$ be the set of vertices on which there was a recursive call
in Line~\ref{line:QuasiRecursiveCalls}
 and $\kappa_s$ the variable whose value is queried on Line~\ref{line:QuasiQueryTwo}.
According to  Line~\ref{line:QuasiQueryTwo} a $0$
is returned if and only if $\sigma(\kappa_s) = 0$ and for every 
$u \in R$ there was at least one recursive call with $\Phi_u$ and distance parameter $\twicelocaldist{\epsilon}$ that returned a $0$.
By the induction assumption this implies that $\sigma\not\in \SAT(\Phi_u)$
for every $u \in R$.
Note that the set $R$ satisfies the exact same conditions that the set $T$ of special relatives satisfies in  Lemma~\ref{prop:witness}.
Hence, Lemma~\ref{prop:witness} asserts that $\sigma\not\in \SAT(\Phi)$.
\end{proof}

We now turn to proving soundness. This depends on first noting that the algorithm will indeed check the paths leading
to critical vertices.

\begin{observation}\label{ob:nottoomany}
If the vertex $s$ picked in Line \ref{line:QuasiPick} is $(\epsilon,\sigma)$-critical, then it will not trigger the condition of Line \ref{line:QuasiTooMany}.
\end{observation}

\begin{proof}
Definition \ref{def:important-critical} in particular implies that for every $u\in A$ (as per Line \ref{line:ancestors}) we have $|\CH{u}|\leq (3/2\epsilon)(1+2\epsilon/3)^{-\lfloor \depth{\Phi}{u}/3\rfloor}\leq 3/2\epsilon$, as otherwise $\sigma$ will be too close to satisfying $\Phi_u$. Also, from Observation \ref{prop:important-notdeep} we know that $\depth{\Phi}{s}\leq 4\mdepth{\epsilon}$ and so $|A|\leq 2\mdepth{\epsilon}+1$.

The two together give us the bound $|R|\leq(3/2\epsilon-1)(2\mdepth{\epsilon}+1)\leq 3\epsilon^{-2}\log (2\epsilon^{-1})$, and so the condition in Line \ref{line:QuasiPick} is not triggered.
\end{proof}

\begin{lemma}~\label{lem:QuasiSoundness}
Let $\sigma$ be $\epsilon$-far from $\SAT(\Phi)$.
If Algorithm~\ref{alg:MonQuasi} is called with $\epsilon$, $\Phi$ and an oracle to $\sigma$, then it returns $0$ with probability at least $\hhitprob{\epsilon}$.
\end{lemma}

\begin{proof}
The base case, $\kappa_r\in X$, is handled correctly by Line \ref{line:QuasiAVariable}.
Assume next that $\epsilon> 3/4$.
Assume that the vertex $s$ selected in  Line~\ref{line:QuasiPick}  is $(\epsilon,\sigma)$-critical.
Hence by definition
$\sigma$ is more than $1/2$-far from $\SAT(\Phi_u)$ for every ancestor $u$ of $s$.
Thus by Observation~\ref{ob:or} we have that $\kappa_u\equiv\wedge$ for every
ancestor $u$ of $s$.
Consequently, by Line~\ref{line:QuasiQueryOne} and Line~\ref{line:QuasiQueryTwo} the value returned 
will be $\sigma(\kappa_s)$, and $\sigma(\kappa_s) = 0$ because $s$ is $(\epsilon,\sigma)$-critical.
By Lemma~\ref{lem:many-critical}, with probability at least $3/16$ the vertex $s$ selected in Line~\ref{line:QuasiPick} is $(\epsilon,\sigma)$-critical.
Thus, $0$ is returned with probability at least $3/16$,
which is greater than $\hhitprob{\epsilon}$ when $3/4<\epsilon\leq 1$.

For all other $\epsilon$ we proceed by induction over the depth. Assume that for any formula $\Phi'$ such that $\depth{}{\Phi'}<\depth{}{\Phi}$ and any assignment $\sigma'$ to $\Phi'$ that is $\eta$-far from $\SAT(\Phi')$ (for any $\eta$), Algorithm \ref{alg:MonQuasi} returns $0$ with probability at least $\hhitprob{\eta}$. Given this we prove that $0$ is returned with probability at least $\hhitprob{\epsilon}$ for $\Phi$ and $\sigma$.

Assume first that the vertex $s$ selected in  Line~\ref{line:QuasiPick} is $(\epsilon,\sigma)$-critical.
Let $A,R$ be the sets from Line~\ref{line:QuasiR} and Line~\ref{line:QuasiA}.
Since $s$ is $(\epsilon,\sigma)$-critical, by definition for every  $u\in A$ we have that $\sigma$ is $\localdist{\epsilon}$-far from $\SAT(\Phi_u)$.
Also, because $s$ is $(\epsilon,\sigma)$-critical, by definition for every  $u\in A$ and $w\in \CH{u}\cap R$ we have that $w\neq \LC{u}$, and therefore
 by Observation~\ref{ob:or} we have that $\sigma$ 
is $\twicelocaldist{\epsilon}$-far from $\SAT(\Phi_w)$ for every $w\in R$.
By the induction assumption, for every $w\in R$,
with probability at least $1-\hhitprob{(\twicelocaldist{\epsilon})}$
Algorithm~\ref{alg:MonQuasi} returns $0$ when called with $\twicelocaldist{\epsilon}$, $\Phi_w$ and an oracle to $\sigma$.
Hence, for every $w\in R$, the probability that on $\orconst{\epsilon}$ such independent calls to Algorithm~\ref{alg:MonQuasi}
the value $0$ was never returned is at
most $(1-\hhitprob{(\twicelocaldist{\epsilon})})^{\orconst{\epsilon}}$.
This is less than $(\numrel{\epsilon})/6$.
Observation \ref{ob:nottoomany} ensures that $|R| \leq 3\numrel{\epsilon}$, and in particular the condition in Line \ref{line:QuasiTooMany} is not invoked and the calls in Line \ref{line:QuasiRecursiveCalls} indeed take place.
By the union bound over the vertices of $R$, with probability at least $1/2$, for every $u\in R$ at least one of calls to Algorithm~\ref{alg:MonQuasi} with $\twicelocaldist{\epsilon}$, $\Phi_u$ and an oracle to $\sigma$ returned the value $0$.
This means that for every $u\in R$, $y_u$ in Line~\ref{line:QuasiRecursiveCalls} was set to $0$ and remained $0$.
Consequently this is the value returned in Line~\ref{line:QuasiQueryTwo}

Finally, since $\sigma$ is $\epsilon$-far from $\SAT(\Phi)$,
by Lemma~\ref{lem:many-critical} the vertex $s$ selected in  Line~\ref{line:QuasiPick} is $(\epsilon,\sigma)$-critical
with probability at least $\hitprob{\epsilon}$.
Therefore $0$ is returned with probability at least $\hhitprob{\epsilon}$.
\end{proof}

\section{The Computational Complexity of the Testers and Estimator}\label{sec:running}

There are two parts to analyzing the computational complexity of a test for a massively parametrized property. The first part is the running time of the preprocessing phase, which reads the entire parameter part of the input, in our case the formula, but has no access yet to the tested part, in our case the assignment. This part is subject to traditional running time and working space definitions, and ideally should have a running time that is quasi-linear or at least polynomial in the size of its input (the ``massive parameter''). The second part is the testing part, which ideally should take a time that is logarithmic in the input size for every query it makes (as a very basic example, even a tester that just makes uniformly random queries over the input would require such a time to draw the necessary $\log (n)$ random coins for each query).

In our case, the preprocessing part would need to take a k-ary formula and convert it to the basic form corresponding to the algorithm that we run. We may assume that the formula is represented as a graph with additional information stored in the vertices.

Constructing the basic form by itself can be done very efficiently (and also have an output size linear in its input size). For example, if the input formula has only ``$\wedge$'' and ``$\vee$'' gates, then a Depth First Search over the input would do nicely, where the output would follow this traversal, but create a new child gate in the output only when it is different than its parent (otherwise it would continue traversing the input while remaining in the same output node). With more general monotone gates, a first pass would convert them to unforceable gates by ``splitting off'' forceful children as in the proof of Lemma \ref{lem:make-kx-basic}. It is not hard to efficiently handle ``$\neg$'' gates using De-Morgan's law too.

Aside from the basic form of the formula, the preprocessing part should construct several additional data structures to make the second part (the test itself) as efficient as possible.

For Algorithm~\ref{alg:MosGen}, we would need to quickly pick a child of a vertex with probability proportional to its sub-formula size, and know who are the light children as well as what is the relative size of the smallest child. This mainly requires storing the size of every sub-formula for every vertex of the tree, as well as sorting the children of each vertex by their sizes and storing the value of the corresponding ``$\ell$''. Algorithm~\ref{alg:GenEst} requires very much the same additional data as Algorithm~\ref{alg:MosGen}. This information can be stored in the vertices of the graph while performing a depth-first traversal of it, starting at the root, requiring a time linear in the size of the basic formula.

For Algorithm~\ref{alg:MonQuasi}, we would need to navigate the tree both downwards and upwards (for finding the ancestors of a vertex), as well as the ability to pick a vertex corresponding to a variable at random, which in itself does not require special preprocessing but does require generating a list of all such vertices. Constructing the set of ancestors is simply following the path from the vertex to the root, requiring time linear in the depth of the vertex in the tree.

The only part in the algorithms above that depends on $\epsilon$ is designating the light children, but this can also be done ``for all $\epsilon$'' at a low cost by storing the range of $\epsilon$ for every positive $\ell$. Since $\ell$ is always an integer no larger than $k+1$, this requires an array of such size in every vertex.

Let us turn to analyzing the running time complexity of the second part, namely the testing algorithm. Once the above preprocessing is performed, the time per instantiation (and thus per query) of the algorithm will be very small (where we charge the time it takes to calculate a recursive call to the recursive instantiation). In Algorithm~\ref{alg:MosGen}, the cost in every instantiation is at most the cost of selecting a child vertex at random for each iteration of the loop in line~\ref{line:MGIsNegloop} and a cost linear in $k$. This would make it a cost logarithmic in the input size per query (multiplied by the time it takes to write and read an address) -- where the log incurrence is in fact only when we need to randomly choose a child according to its weight. The case of Algorithm~\ref{alg:GenEst} is similar, except that there is also a cost for every term in the $\mDNF$, of which there are at most $2^k$.

For Algorithm~\ref{alg:MonQuasi}, every instantiation requires iterating over all the ancestors of one vertex picked at random. This requires time linear in the depth of the formula and logarithmic in the input size per query.

\section{The Untestable Formulas}\label{sec:untest}

We describe here a read-once formula over an alphabet with $4$ values, defining a property that cannot be $1/4$-tested using a constant number of queries. The formula will have a very simple structure, with only one gate type. Then, building on this construction, we describe a read-once formula over an alphabet with $5$-values that cannot be $1/12$-tested, which satisfies an additional monotonicity condition: All gates as well as the acceptance condition are additionally monotone with respect to a fixed ordering of the alphabet.

\subsection{The $4$-valued formula}

For convenience we denote our alphabet by $\Sigma=\{0,1,P,F\}$. An input is said to be {\em accepted} by the formula if, after performing the calculations in the gates, the value received at the root of the tree is not ``$F$''. We restrict the input variables to $\{0,1\}$, although it is easy to see that the following argument holds also if we allow other values to the input variables (and also if we change the acceptance condition to the value at the root having to be ``$P$'').

\begin{definition}
The {\em balancing gate} is the gate that receives two inputs from $\Sigma$ and outputs the following.
\begin{itemize}
\item For $(0,0)$ the output is $0$ and for $(1,1)$ the output is $1$.
\item For $(1,0)$ and $(0,1)$ the output is $P$.
\item For $(P,P)$ the output is $P$,
\item For anything else the output is $F$.
\end{itemize}

For a fixed $h>0$, the {\em balancing formula} of height $h$ is the formula defined by the following.
\begin{itemize}
\item The tree is the full balanced binary tree of height $h$ with variables at the leaves, and hence there are $2^h$ variables.
\item All gates are set to the balancing gate.
\item The formula accepts if the value output at the root is not ``$F$''.
\end{itemize}
\end{definition}

We denote the variables of the formula in their order by $x_0,\ldots,x_{2^h-1}$. The following is easy.

\begin{lemma}\label{lem:baliff}
An assignment $a_0\in\{0,1\},\ldots,a_{2^h-1}\in\{0,1\}$ to $x_0,\ldots,x_{2^h-1}$ is accepted by the formula if and only if for every $0<k\leq h$ and every $0\leq i<2^{h-k}$, the number of $1$ values in $a_{i2^k},\ldots,a_{(i+1)2^k-1}$ is either $0$, $2^k$ or $2^{k-1}$.
\end{lemma}

\begin{proof}
Denote the number of $1$ values in variables descending from a gate $u$ by $\mathrm{num}_1(u)$. Let us prove by induction on $k$ that:
\begin{itemize}
\item $\mathrm{num}_1(v) =0$ if and only if the value of $v$ is $0$,
\item $\mathrm{num}_1(v) =2^k$ if and only if the value of $v$ is $1$,
\item and $\mathrm{num}_1(v) =2^{k-1}$ if and only if the value of $v$ is $P$.
\end{itemize} 

For $k=1$ we have the two inputs of $v$, and by the definition of the balancing gate the claim follows. 

For $k>1$, we have $2^k$ variables which are all descendants of the same gate $v$. By the induction hypothesis, for both children of $v$, denoted $u,w$, we have that $\mathrm{num}_1(u),\mathrm{num}_1(w)\in \{0,2^{k-2},2^{k-1}\}$ and that this determines their value (unless at least one of them already evaluates to $F$, in which case both the entire formula is not satisfied, and by induction there is an interval without the correct number of $1$ values). If $\mathrm{num}_1(w) =\mathrm{num}_1(u) =0$ then they both evaluate to $0$ and so does $v$. Similarly, if $\mathrm{num}_1(w) =\mathrm{num}_1(u) =2^{k-1}$ then both evaluate to $1$ and so does $v$. If $\mathrm{num}_1(u)=2^{k-1}$ and $\mathrm{num}_1(w)=0$, then $u$ evaluates to $1$ and $w$ to $0$ and indeed $v$ evaluates to $P$ (and similarly for the symmetric case). If  $\mathrm{num}_1(u)= \mathrm{num}_1(w)=2^{k-2}$, then both evaluate to $P$ and so does $v$. The remaining case is $\mathrm{num}_1(u) \in \{0,2^{k-1}\}$ and $\mathrm{num}_1(w)=2^{k-2}$ (and the symmetric case), by the induction hypothesis and the definition of the balancing gate this implies that $v$ evaluates to $F$ and the formula is unsatisfied.\end{proof}

In other words, every ``binary search interval'' is either all $0$, or all $1$, or has the same number of $0$ and $1$. This will allow us to easily prove that certain inputs are far from satisfying the property.

\subsection{Two distributions}

We now define two distributions, one over satisfying inputs and the other over far inputs.

\begin{definition}
The distribution $D_Y$ is defined by the following process.
\begin{itemize}
\item Uniformly pick $2\leq k\leq h$.
\item For every $0\leq i<2^{h-k}$, independently pick either $(y_{i,0},y_{i,1})=(0,1)$ or $(y_{i,0},y_{i,1})=(1,0)$ (each with probability $1/2$).
\item For every $0\leq i<2^{h-k}$, set $$x_{i2^k}=\cdots=x_{i2^k+2^{k-1}-1}=y_{i,0};\qquad x_{i2^k+2^{k-1}}=\cdots=x_{(i+1)2^k-1}=y_{i,1}.$$
\end{itemize}
\end{definition}

\begin{definition}
The distribution $D_N$ is defined by the following process.
\begin{itemize}
\item Uniformly pick $2\leq k\leq h$.
\item For every $0\leq i<2^{h-k}$, independently choose $(z_{i,0},z_{i,1},z_{i,2},z_{i,3})$ to have either one $1$ and three $0$ or one $0$ and three $1$ (each of the $8$ possibilities with probability $1/8$).
\item For every $0\leq i<2^{h-k}$, set $$x_{i2^k}=\cdots=x_{i2^k+2^{k-2}-1}=z_{i,0};\qquad x_{i2^k+2^{k-2}}=\cdots=x_{i2^k+2^{k-1}-1}=z_{i,1};$$ $$x_{i2^k+2^{k-1}}=\cdots=x_{i2^k+2^{k-1}+2^{k-2}-1}=z_{i,2};\qquad x_{i2^k+2^{k-1}+2^{k-2}}=\cdots=x_{(i+1)2^k-1}=z_{i,3}.$$
\end{itemize}
\end{definition}

It is easier to illustrate this by considering the calculation that results from the distributions. In both distributions we can think of a randomly selected level $k$ (counted from the bottom, where the leaf level $0$ and the level above it $1$ are never selected). In $D_Y$, the output of all gates at or above level $k$ is ``$P$'', while the inputs to every gate at level $k$ will be either $(0,1)$ or $(1,0)$, chosen uniformly at random.

In $D_N$ all gates at level $k$ will output ``$F$'' (note however that we cannot query a gate output directly); looking two levels below, every gate as above holds the result from a quadruple chosen uniformly from the $8$ choices described in the definition of $D_N$ (the quadruple $(z_{i,0},z_{i,1},z_{i,2},z_{i,3})$).  At level $k-2$ or lower the gate outputs are $0$ and $1$ and their distribution resembles very much the distribution as in the case for $D_Y$ -- as long as we cannot ``focus'' on the transition level $k$. This is formalized in terms of lowest common ancestors below.

\begin{lemma}\label{lem:missdiff}
Let $Q\subset\{1,\ldots,2^h\}$ be a set of queries, and let $H\subset\{0,\ldots,h\}$ be the set of levels containing lowest common ancestors of subsets of $Q$. Conditioned on neither $k$ nor $k-1$ being in $H$, both $D_Y$ and $D_N$ induce exactly the same distribution over the outcome of querying $Q$.
\end{lemma}

\begin{proof}
Let us condition the two distributions on a specific value of $k$ satisfying the above. For two queries $q,q'\in Q$ whose lowest common ancestor is on a level below $k-1$, with probability $1$ they will receive the exact same value (this holds for both $D_N$ and $D_Y$). The reason is clear from the construction -- their values will come from the same $y_{i,j}$ or $z_{i,j}$.

Now let $Q'$ contain one representative from every set of queries in $Q$ that must receive the same value by the above argument. For any $q,q'\in Q'$, their lowest common ancestor is on a level above $k$. For $D_Y$ it means that $x_q$ takes its value from some $y_{i,j}$ and $x_{q'}$ takes its value from some $y_{i',j'}$ where $i\neq i'$. Because each pair $(y_{i,0},y_{i,1})$ is chosen independently from all others, this means that the outcome of the queries in $Q'$ is uniformly distributed among the $2^{|Q'|}$ possibilities. The same argument (with $z_{i,j}$ and $z_{i',j'}$ instead of $y_{i,j}$ and $y_{i',j'}$) holds for $D_N$. Hence the distribution of outcomes over $Q'$ is the same for both distributions, and by extension this holds over $Q$.
\end{proof}

On the other hand, the two distributions are very different with respect to satisfying the formula.

\begin{lemma}\label{lem:satnfar}
An input chosen according to $D_Y$ always satisfies the balancing formula, while an input chosen according to $D_N$ is always $1/4$-far from satisfying it.
\end{lemma}

\begin{proof}
By Lemma \ref{lem:baliff}, the assignment constructed in $D_Y$ will always be satisfied. This is since for every vertex in a level lower than $k$, all of its descendant variables will be of the same value, and for every vertex in level $k$ or above exactly half of the variables will have each value.

Note that in an input constructed according to $D_N$, every vertex at level $k$ has one quarter of its descendant variables of one value, while the rest are of the other one. By averaging, if one were to change less than $1/4$ of the input values, we will have one vertex $v$ at level $k$ for which less than $1/4$ of the values of its descendant variables were changed. This means that $v$ cannot satisfy the requirements in Lemma \ref{lem:baliff} and therefor it, and hence the entire formula, evaluate to $F$.
\end{proof}

\subsection{Proving non-testability}

We use here the following common application of Yao's method (see e.g. \cite{Fischer01theart}).

\begin{lemma}\label{lem:adaptyao}
If $D_Y$ is a distribution over satisfying inputs and $D_N$ is a distribution over $\epsilon$-far ones, such that for any fixed set of queries $Q$ with $|Q|\leq l$ the probability distributions over the outcomes differ by less than $\frac13$ (in the variation distance norm) for $D_Y$ and $D_N$, then there is no non-adaptive $\epsilon$-test for the property that makes at most $l$ queries ($1$-sided or $2$-sided).
\end{lemma}

This allows us to conclude the proof.

\begin{theorem}\label{thm:untest}
Testing for being a satisfying assignment of the balancing formula of height $h$ requires at least $\Omega(h)$ queries for a non-adaptive test and $\Omega(\log h)$ queries for a possibly adaptive one.
\end{theorem}

\begin{proof}
We note that for any set of queries $Q$, the size of the set of lowest common ancestors (outside $Q$ itself) is less than $Q$, and hence (in the notation of Lemma \ref{lem:missdiff}) we have $|H|\leq|Q|$. Now if $|Q|=o(h)$, then the event of Lemma \ref{lem:missdiff} happens with probability $1-o(1)$, and hence the variation distance between the two (unconditional) distributions over outcomes is $o(1)$. Together with Lemma \ref{lem:satnfar} this fulfills the conditions for Lemma \ref{lem:adaptyao} for concluding the proof.

For adaptive algorithms the bound follows by the standard procedure that makes an adaptive algorithm into a non-adaptive one at an exponential cost (by querying in advance the algorithm's entire decision tree given its internal coin tosses).
\end{proof}

\subsection{An untestable $5$-valued monotone formula}

While the lower bound given above uses a gate which is highly non-monotone, we can also give a similar construction where the alphabet is of size $5$ and the gates are monotone (that is, where increasing any input of the gate according to the order of the alphabet does not decrease its input).

Instead of just ``$\{1,\ldots,5\}$'' we denote our alphabet by $\Sigma=\{0,F_0,P,F_1,1\}$ in that order.
We will restrict the input variables to $\{0,1\}$, although it is not hard to generalize to the case where the input variables may take any value in the alphabet. At first we analyze a formula that has a non-monotone satisfying condition.

\begin{definition}
The {\em monotone balancing gate} is the gate that receives two inputs from $\Sigma$ and outputs the following.
\begin{itemize}
\item For $(0,0)$ the output is $0$ and for $(1,1)$ the output is $1$.
\item For $(1,0)$ and $(0,1)$ the output is $P$.
\item For $(P,P)$ the output is $P$.
\item For $(0,P)$ and $(P,0)$ the output is $F_0$.
\item For $(1,P)$ and $(P,1)$ the output is $F_1$.
\item For $(P,F_0)$, $(F_0,P)$, $(F_0,0)$, $(0,F_0)$ and $(F_0,F_0)$ the output is $F_0$.
\item For $(F_0,1)$ and $(1,F_0)$ the output is $F_1$.
\item For any pair of inputs containing $F_1$, the output is $F_1$.
\end{itemize}
For a fixed $h>0$, the {\em almost-monotone balancing formula} of height $h$ is the formula defined by the following.
\begin{itemize}
\item The tree is the full balanced binary tree of height $h$ with variables at the leaves, and hence there are $2^h$ variables.
\item All gates are set to the monotone balancing gate.
\item The formula accepts if the value output at the root is not ``$F_0$'' or ``$F_1$''.
\end{itemize}
\end{definition}

The following observation is easy by just running over all possible outcomes of the gate.

\begin{observation}\label{ob:isomon}
The monotone balancing gate in monotone. Additionally, if the values $F_0$ and $F_1$ are unified then the gate is still well-defined, and is isomorphic to the $4$-valued balancing gate.
\end{observation}

In particular, the above observation implies that the almost-monotone balancing formula has the same property testing lower bound as that of the balancing formula, using the same proof with the same distributions $D_Y$ and $D_N$. However, we would like a completely monotone formula. For that we use a monotone decreasing acceptance condition; we note that a formula with a monotone increasing acceptance condition can be obtained from it by just ``reversing'' the order over the alphabet.

\begin{definition}
The {\em monotone sub-balancing formula} is defined the same as the almost-monotone balancing formula, with the exception that the formula accepts if and only if the value output at the root is not $F_1$ or $1$.
\end{definition}

By Observation \ref{ob:isomon}, the distribution $D_Y$ is also supported by inputs satisfying the monotone sub-balancing formula. To analyze $D_N$, note the following.

\begin{lemma}\label{lem:subbal}
An assignment $a_0\in\{0,1\},\ldots,a_{2^h-1}\in\{0,1\}$ to $x_0,\ldots,x_{2^h-1}$ for which for some $0<k\leq h$ and some $0\leq i<2^{h-k}$, the number of $1$ values in $a_{i2^k},\ldots,a_{(i+1)2^k-1}$ is more than $2^{k-1}$ and less than $2^k$ cannot be accepted by the formula.
\end{lemma}

\begin{proof}
We set $u$ to be the gate whose descendant variables are exactly $a_{i2^k},\ldots,a_{(i+1)2^k-1}$. We first note that it is enough to proof that $u$ evaluates to $F_1$, because then by the definition of the gates the root will also evaluate to $F_1$. We then use induction over $k$, while referring to Observation \ref{ob:isomon} and the proof of Lemma \ref{lem:baliff}. The base case $k=1$ is true because then no assignment satisfies the conditions of the lemma.

If any of the two children of $u$ evaluates to $F_1$ the we are also done by the definition of the gate. The only other possible scenario (using induction) is when one of the children $v$ of $u$ must evaluate to $1$, and hence all of its $2^{k-1}$ descendant variables are $1$, while for the other child $w$ of $u$ some of the descendant variables are $0$ and some are $1$.  But this means that $w$ does not evaluate to either $0$ or $1$, which again means that $u$ evaluates to $F_1$.
\end{proof}

This yields the following.

\begin{lemma}\label{lem:subfar}
With probability $1-o(1)$, an input chosen according to $D_N$ will be $1/12$-far from satisfying the monotone sub-balancing formula.
\end{lemma}

\begin{proof}
This is almost immediate from Lemma \ref{lem:subbal}, as a large deviation inequality implies that with probability $1-o(1)$, more than $1/3$ of the quadruples $(z_{i,0},z_{i,1},z_{i,2},z_{i,3})$ as per the definition of $D_N$ will have three $1$'s and one $0$.
\end{proof}

Now we can prove a final lower bound.

\begin{theorem}\label{thm:untestmon}
Testing for being a satisfying assignment of the monotone sub-balancing formula of height $h$ requires at least $\Omega(h)$ queries for a non-adaptive test and $\Omega(\log h)$ queries for a possibly adaptive one.
\end{theorem}

\begin{proof}
This follows exactly the proof of the lower bound for the balancing formula. Due to Observation \ref{ob:isomon} and Lemma \ref{lem:subfar} we can use the same $D_Y$ and $D_N$ (the $o(1)$ probability of $D_N$ not producing a far input makes no essential difference for the use of Yao's method).
\end{proof}

\bibliographystyle{plain}
\bibliography{TF}




\end{document}